\documentclass[a4paper,11pt]{article}
\usepackage{fullpage}

\usepackage{color}
\usepackage{amsmath,amssymb,enumerate}
\usepackage{algorithmic}
\usepackage{algorithm}
\usepackage{amsthm}
\usepackage{graphicx}
\usepackage{sidecap}
\usepackage{environ}
\sidecaptionvpos{figure}{t}
\usepackage{subcaption}
\usepackage{url}
\usepackage{hyperref}
\usepackage{comment}
\hypersetup{colorlinks=true}
\hypersetup{linkcolor=black,anchorcolor=black,citecolor=black,urlcolor=black}
\hypersetup{pdfauthor= {Dimitris Fotakis}{Loukas Kavouras}{Grigorios Koumoutsos}{Stratis Skoulakis} {Manolis Vardas}} 
\hypersetup{pdfkeywords={online algorithms}{competitive analysis}{min-sum set cover}{multiple intents re-ranking}{web search ranking}} 
\hypersetup{pdftitle=The Online min-sum Set Cover problem.}

\usepackage[algo2e,ruled]{algorithm2e}
\usepackage{cite}

\usepackage{thmtools}
\usepackage{thm-restate}
\usepackage{cleveref}

\newtheorem{theorem}{Theorem}[section]
\newtheorem{definition}[theorem]{Definition}
\newtheorem{lemma}[theorem]{Lemma}

\newcounter{note}[section]

\renewcommand{\thenote}{\thesection.\arabic{note}}

\newcommand{\lknote}[1]{\stepcounter{note}\textcolor{blue}{$\ll${\bf Loukas~\thenote:} {\sf #1}$\gg$\marginpar{\tiny\bf lk~\thenote}}}

\newcommand{\cA}{\ensuremath{\mathcal{A}}}

\newcommand{\cE}{\ensuremath{\mathcal{E}}}
\newcommand{\cF}{\ensuremath{\mathcal{F}}}

\newcommand{\cL}{\ensuremath{\mathcal{L}}}

\newcommand{\cS}{\ensuremath{\mathcal{S}}}

\newcommand{\repeattheorem}[1]{%
  \begingroup
  \renewcommand{\thetheorem}{\ref{#1}}%
  \expandafter\expandafter\expandafter\theorem
  \csname reptheorem@#1\endcsname
  \endtheorem
  \endgroup
}

\NewEnviron{reptheorem}[1]{%
  \global\expandafter\xdef\csname reptheorem@#1\endcsname{%
    \unexpanded\expandafter{\BODY}%
  }%
  \expandafter\theorem\BODY\unskip\label{#1}\endtheorem
}

\DeclareMathOperator{\acost}{AccessCost}

\DeclareMathOperator{\ALG}{ALG}
\DeclareMathOperator*{\argmin}{arg\,min}
\DeclareMathOperator*{\expect}{\mathbb{E}}
\DeclareMathOperator{\cost}{Cost}
\DeclareMathOperator{\costmax}{C_{\max}}
\DeclareMathOperator{\deterministic}{Derand}

\DeclareMathOperator{\kt}{KT}

\DeclareMathOperator{\LRA}{Lazy-Rounding}
\DeclareMathOperator{\MAE}{MAE}
\DeclareMathOperator{\mcost}{MovingCost}
\DeclareMathOperator{\MTF}{MTF}
\DeclareMathOperator{\MWU}{MWU}
\DeclareMathOperator{\OPT}{OPT}

\DeclareMathOperator{\phase}{start-phase}
\DeclareMathOperator{\pc}{\textbf{c}}

\DeclareMathOperator*{\Prob}{Pr}

\DeclareMathOperator{\tv}{TV}

\makeatletter
\newcommand*{\rom}[1]{\expandafter\@slowromancap\romannumeral #1@}

\def\hf{\selectfont\sffamily\bfseries}

\newcommand{\lrp}[1]{\left( #1 \right)}

\renewcommand{\paragraph}[1]{\vspace{0.14cm} \noindent \textbf{#1}}

\title{The Online Min-Sum Set Cover Problem}

\author{
Dimitris Fotakis\thanks{National Technical University of Athens, Greece.  \texttt{fotakis@cs.ntua.gr,lukaskavouras@gmail.com}. Dimitris Fotakis is supported by the Hellenic Foundation for Research and Innovation (H.F.R.I.) under the ``First Call for H.F.R.I. Research Projects to support Faculty members and Researchers' and the procurement of high-cost research equipment grant'', project BALSAM, HFRI-FM17-1424.}\qquad
Loukas Kavouras\footnotemark[1]\thanks{Supported by a scholarship from the State
Scholarships Foundation,
co-financed by Greece and the European Union (European Social Fund-ESF).}
\qquad
Grigorios Koumoutsos\thanks{Universit\'{e} libre de Bruxelles, Belgium. \texttt{gregkoumoutsos@gmail.com}. Supported by Fonds de la Recherche Scientifique-FNRS Grant no MISU F 6001. Part of this work was carried out while visiting that National Technical University of Athens (NTUA), supported by FNRS Mobility Grant no 35282070.}\\[1.2ex] 
\qquad
Stratis Skoulakis\thanks{Singapore University of Technology and Design. \texttt{efstratios@sutd.edu.sg}. Supported by NRF 2018 Fellowship NRF-NRFF2018-07. Part of this work was carried out while  the author was a PhD student at NTUA.}\qquad
Manolis Vardas\thanks{ETH Zurich. \texttt{evardas@student.ethz.ch}. This research was carried out while the author was an undergraduate student at NTUA. }
}

\title{The Online Min-Sum Set Cover Problem}

\begin{document}
\date{}
\maketitle
\begin{abstract}
We consider the online Min-Sum Set Cover (MSSC), a natural and intriguing generalization of the classical list update problem. In Online MSSC, the algorithm maintains a permutation on $n$ elements based on subsets $S_1, S_2, \ldots$ arriving online. The algorithm serves each set $S_t$ upon arrival, using its current permutation $\pi_{t}$, incurring an access cost equal to the position of the first element of $S_t$ in $\pi_{t}$. Then, the algorithm may update its permutation to $\pi_{t+1}$, incurring a moving cost equal to the Kendall tau distance of $\pi_{t}$ to $\pi_{t+1}$. The objective is to minimize the total access and moving cost for serving the entire sequence. We consider the $r$-uniform version, where each $S_t$ has cardinality $r$. List update is the special case where $r = 1$.

We obtain tight bounds on the competitive ratio of deterministic online algorithms for MSSC against a static adversary, that serves the entire sequence by a single permutation. First, we show a lower bound of $(r+1)(1-\frac{r}{n+1})$ on the competitive ratio. Then, we consider several natural generalizations of successful list update algorithms and show that they fail to achieve any interesting competitive guarantee. On the positive side, we obtain a $O(r)$-competitive deterministic algorithm using ideas from online learning and the multiplicative weight updates (MWU) algorithm.

Furthermore, we consider efficient algorithms. We propose a memoryless online algorithm, called \emph{Move-All-Equally}, which is inspired by the Double Coverage algorithm for the $k$-server problem. We show that its competitive ratio is $\Omega(r^2)$ and $2^{O(\sqrt{\log n  \cdot \log r})}$, and conjecture that it is $f(r)$-competitive. We also compare Move-All-Equally against the dynamic optimal solution and obtain (almost) tight bounds by showing that it is $\Omega(r \sqrt{n})$ and $O(r^{3/2} \sqrt{n})$-competitive.
\end{abstract}

\newpage
\setcounter{page}{1}
\includecomment{onlymain}
\excludecomment{onlyapp}

\begin{onlymain}
\section{Introduction}

In Min-Sum Set Cover (MSSC), we are given a universe $U$ on $n$ elements and a collection of subsets $\cS = \lbrace S_1, \dotsc, S_m \rbrace$, with $S_t \subseteq U$, and the task is to construct a permutation (or list) $\pi$ of elements of $U$. The cost $\pi(S_t)$ of covering a set $S_t$ (a.k.a. the cover time of $S_t$) with a permutation $\pi$ is the position of the first element of $S_t$ in $\pi$, i.e., $\pi(S_t) = \min \lbrace i \,|\, \pi(i) \in S_t \rbrace $. The goal is to minimize the overall cost $\sum_{t} \pi(S_t)$ of covering all subsets of $\cS$.

The MSSC problem generalizes various NP-hard problems such as Min-Sum Vertex Cover and Min-Sum Coloring and it is well-studied. Feige, Lovasz and Tetali~\cite{FLT04} showed that the greedy algorithm, which picks in each position the element that covers the most uncovered sets, is a 4-approximation (this was also implicit in~\cite{BBHST98}) and that no $(4-\epsilon)$-approximation is possible, unless $\text{P} = \text{NP}$. 
Several generalizations have been considered over the years with applications in various areas (we discuss some of those problems and results in Section~\ref{sec:related_work}).

\paragraph{Online Min-Sum Set Cover.} In this paper, we study the online version of Min-Sum Set Cover. Here, the sets arrive online; at time step $t$, the set $S_t$ is revealed. An online algorithm is charged the \emph{access cost} of its current permutation $\pi_t(S_t)$; then, it is allowed to change its permutation to $\pi_{t+1}$ at a \emph{moving cost} equal to the number of inversions between $\pi_t$ and $\pi_{t+1}$, known as the Kendall tau distance $d_{\kt}(\pi_{t},\pi_{t+1})$. The goal is to minimize the total cost, i.e., $\sum_{t} \big( \pi_t(S_t) + d_{\kt}(\pi_{t},\pi_{t+1}) \big)$. This is a significant generalization of the classic list update problem, which corresponds to the special case where $|S_t|=1$ for all sets $S_t \in \cS$.




\paragraph{Motivation.} Consider a web search engine, such as Google. Each query asked might have many different meanings depending on the user. For example, the query ``Python'' might refer to an animal, a programming language or a movie. Given the pages related to ``Python'', a goal of the search engine algorithm is to rank them such that for each user, the pages of interest appear as high as possible in the ranking (see e.g., \cite{DKNS01}). 
Similarly, news streams include articles covering different reader interests each. We want to rank the articles so that every reader finds an article of interest as high as possible.
The MSSC problem serves as a theoretical model for practical problems of this type, 
where we want to aggregate disjunctive binary preferences (expressed by the input sets) into a total order.  
E.g., for a news stream, the universe $U$ corresponds to the available articles and the sets $S_t$ correspond to different user types. The cost of a ranking (i.e., permutation on $U$) for a user type is the location of the first article of interest. Clearly, in such applications, users arrive online and the algorithm might need to re-rank the stream (i.e., change the permutation) based on user preferences.

\paragraph{Benchmarks.} 
For the most part, we evaluate the performance of online algorithms by comparing their cost against the cost of an optimal offline solution that knows the input in advance and chooses an optimal permutation $\pi$. Note that this solution is \textit{static}, in the sense that it does not change permutations over time. This type of analysis, called \textit{static optimality}, is typical in online optimization and online learning. It was initiated in the context of adaptive data structures by the landmark result of Sleator and Tarjan~\cite{ST85b}, who showed that \textit{splay trees} are asymptotically as fast as any \textit{static} tree. Since then, it has been an established benchmark for various problems in this area (see e.g.~\cite{IM12,BCK03}); it is also a standard benchmark for several other problems in online optimization (e.g., online facility location~\cite{F08,Meyerson01}, minimum metric matching~\cite{GGPW19,KN03,NR17}, Steiner tree~\cite{NPS11}, etc.). 

A much more general benchmark is the \textit{dynamic Min-Sum Set Cover} problem, where the algorithm is compared against an optimal solution allowed to change permutations over time. This problem has not been studied even in the offline case. In this work, we define the problem formally and obtain first results for the online case.

We remark that the online dynamic MSSC problem belongs to a rich class of problems called \textit{Metrical Task Systems} (MTS)~\cite{BLS92}. MTS is a far-reaching generalization of several fundamental online problems and provides a unified framework for studying online problems (we discuss this in more detail in Section~\ref{sec:related_work}). Indeed, our results suggest that solving the online dynamic MSSC requires the development of powerful generic techniques for online problems, which might have further implications for the broader setting of MTS.

Throughout this paper, whenever we refer to online problems, like Min-Sum Set Cover or list update, we assume the static case, unless stated otherwise. 

\paragraph{Previous Work on List Update.} 
Prior to our work, the only version of online MSSC studied is the special case where $|S_t|=1$ for all sets; this is the celebrated list update problem and it has been extensively studied (an excellent reference is~\cite{BEY98}). It is known that the deterministic competitive ratio it least $2- \frac{2}{n+1}$ and there are several 2-competitive algorithms known; most notably, the Move-to-Front (MTF) algorithm, which moves the (unique) element of $S_t$ to the first position of the permutation, and the Frequency Count algorithm, which orders the elements in decreasing order according to their frequencies. 

The dynamic list update problem has also been extensively studied. MTF is known to be 2-competitive~\cite{ST85} and there are several other 2-competitive algorithms~\cite{Albers98,EY96}. 

\subsection{Our Results}


In this work, we initiate a systematic study of the online Min-Sum Set Cover problem. We consider the $r$-uniform case, where all request sets have the same size $|S_t| = r$. This is without loss of generality, as we explain in Section~\ref{sec:prelim}. 

The first of our main results is a tight bound on the deterministic competitive ratio of Online MSSC. We show that the competitive ratio of deterministic algorithms is $\Omega(r)$.

\begin{theorem}
\label{thm:lb}
Any deterministic online algorithm for the Online Min-Sum Set Cover problem has competitive ratio at least $(r+1) (1 - \frac{r}{n+1})$. 
\end{theorem}

Note that for $r=1$, this bound evaluates  to $2 - \frac{2}{n+1}$, which is exactly the best known lower bound for the list update problem. 

We complement this result by providing a matching (up to constant factors) upper bound. 

\begin{theorem}
\label{thm:static_ub}
There exists a $(5r+2)$-competitive deterministic online algorithm for the Online Min-Sum Set Cover problem.
\end{theorem}

Interestingly, all prior work on the list update problem (case $r=1$) does not seem to provide us with the right tools for obtaining an algorithm with such guarantees! As we discuss in Section~\ref{sec:lower_bounds}, virtually all natural generalizations of successful list update algorithms (e.g., Move-to-Front, Frequency Count) end up with a competitive ratio way far from the desired bound. In fact, even for $r=2$, most of them have a competitive ratio depending on $n$, such as $\Omega(\sqrt{n})$ or even $\Omega(n)$. 

This suggests that online MSSC has a distinctive combinatorial structure, very different from that of list update, whose algorithmic understanding calls for significant new insights. The main reason has to do with the disjunctive nature of the definition of the access cost $\pi(S_t)$. In list update, where $r = 1$, the optimal solution is bound to serve a request $S_t$ by its unique element. The only question is how fast an online algorithm should upgrade it (and the answer is ``as fast as possible''). In MSSC, the hard (and crucial) part behind the design of any competitive algorithm is how to ensure that the algorithm learns fast enough about the element $e_t$ used by the optimal solution to serve each request $S_t$. This is evident in the highly adaptive nature of the deceptively simple greedy algorithm of \cite{FLT04} and in the adversarial request sequences for generalizations of Move-to-Front, in Section~\ref{sec:lower_bounds}. 

To obtain the asymptotically optimal ratio of Theorem~\ref{thm:static_ub}, we develop a rounding scheme and use it to derandomize the multiplicative weights update (MWU) algorithm. 
Our analysis bounds the algorithm's access cost in terms of the optimal cost, but it does not account for the algorithm's moving cost. We then refine our approach, by performing lazy updates to the algorithm's permutation, and obtain a competitive algorithm for online MSSC. 

We also observe (in Section~\ref{sec:prelim}) that based on previous work of Blum and Burch~\cite{BB00}, there exists a (computationally inefficient) randomized algorithm with competitive ratio $1+\epsilon$, for any $\epsilon \in (0, 1/4)$. This implies that no lower bound is possible, if randomization is allowed, and gives a strong separation between deterministic and randomized algorithms. 




\paragraph{Memoryless Algorithms.} 
While the bounds of Theorems~\ref{thm:lb} and ~\ref{thm:static_ub} are matching, our algorithm from Theorem~\ref{thm:static_ub} is computationally inefficient since it simulates the MWU algorithm, which in turn, maintains a probability distribution over all $n!$ permutations. This motivates the study of trade-offs between the competitive ratio and computational efficiency. To this end, we propose a memoryless algorithm, called \emph{Move-All-Equally} ($\MAE$), which moves all elements of set $S_t$ towards the beginning of the permutation at the same speed until the first reaches the first position. This is inspired by the Double Coverage algorithm from $k$-server~\cite{CKPV91,CL91}. We believe that $\MAE$ achieves the best guarantees among all memoryless algorithms. We show that this algorithm can not match the deterministic competitive ratio. 

\begin{reptheorem}{thm:mae-lb} 
 The competitive ratio of the Move-All-Equally algorithm is $\Omega(r^2)$.
\end{reptheorem}

Based on Theorem~\ref{thm:mae-lb}, we conjecture that an $O(r)$ guarantee cannot be achieved by a memoryless algorithms. We leave as an open question whether $\MAE$ has a competitive ratio $f(r)$, or a dependence on $n$ is necessary. To this end, we show that the competitive ratio of $\MAE$ is at most $2^{O(\sqrt{\log n \cdot \log r})}$ (see Section~\ref{sec:algorithm} for details).

\paragraph{Dynamic Min-Sum Set Cover.} 
We also consider the dynamic version of online MSSC. Dynamic MSSC is much more general and the techniques developed for the static case do not seem adequately powerful. This is not surprising, since the MWU algorithm is designed to perform well against the best static solution. 
We investigate the performance of the $\MAE$ algorithm. First, we obtain an upper bound on its competitive ratio.

\begin{reptheorem}{thm:mae-dyn-ub}
 The competitive ratio of the Move-All-Equally algorithm for the dynamic online Min-Sum Set Cover problem is $O(r^{3/2} \sqrt{n})$.
\end{reptheorem}

Although this guarantee is not very strong, we show that, rather surprisingly, it is essentially tight and no better guarantees can be shown for this algorithm.

\begin{reptheorem}{thm:mae-dyn-lb}
For any $r \geq 3$, the competitive ratio of the Move-All-Equally algorithm for the dynamic online Min-Sum Set Cover problem is $\Omega(r \sqrt{n})$.
\end{reptheorem}

 
This lower bound is based on a carefully crafted adversarial instance; this construction reveals the rich structure of this problem and suggests that more powerful generic techniques are required in order to achieve any $f(r)$ guarantees. In fact, we conjecture that the lower bound of Theorem~\ref{thm:lb} is the best possible (ignoring constant factors) even for the dynamic problem and that using a work-function based approach such a bound can be obtained. 


\subsection{Further Related Work}
\label{sec:related_work}


\paragraph{Multiple Intents Re-ranking.} This is a generalization of MSSC where for each set $S_t$, there is a \textit{covering requirement} $K(S_t)$, and the cost of covering a set $S_t$ is the position of the $K(S_t)$-th element of $S_t$ in $\pi$. The MSSC problem is the special case where $K(S_t)=1$ for all sets $S_t$. Another notable special case is the Min-Latency Set Cover problem, which corresponds to the other extreme case where $K(S_t) = |S_t|$~\cite{HL05}. Multiple Intents Re-ranking was first studied by Azar et. al.~\cite{AGY09}, who presented a $O(\log r)$-approximation; later $O(1)$-approximation algorithms were obtained~\cite{BGK10,SW11,ISZ14}.
Further generalizations have been considered, such as the Submodular Ranking problem, studied by Azar and Gamzu~\cite{AG11}, which generalizes both Set Cover and MSSC, and the Min-Latency Submodular Cover, studied by Im et.al~\cite{INZ16}. 

\paragraph{Prediction from Expert Advice and Randomized MSSC.} In prediction from  expert advice, there are $N$ experts and expert $i$ incurs a cost $c_i^t$ in each step. A learning algorithm decides which expert $i_t$ to follow (before the cost vector $\pc^t$ is revealed) and incurs a cost of $c^t_{i_t}$. The landmark technique for solving this problem is the multiplicative weights update (MWU - a.k.a. Hedge) algorithm. For an in-depth treatment of MWU, we refer to~\cite{LW94,FS97,AHK12}. 

In the classic online learning setting, there is no cost for moving probability mass
between experts. However, in a breakthrough result, Blum and Burch~\cite{BB00} showed that MWU is $(1+\epsilon)$-competitive against the best expert, even if there is a cost $D$ for moving probability mass between experts. By adapting this result to online MSSC (regarding permutations as experts), we can get an (inefficient) randomized algorithm with competitive ratio $(1+\epsilon)$, for any constant $\epsilon \in (0,1/4)$. A detailed description is deferred to the full version of this paper.


\paragraph{Metrical Task Systems and Online Dynamic MSSC.} 
The online dynamic Min-Sum Set Cover problem belongs to a rich family of problems called Metrical Task Systems (MTS). In MTS, we are given a set of $N$ states and a metric function $d$ specifying the cost of moving between the states. At each step, a task arrives; the cost of serving the task at state $i$ is $c_i$. An algorithm has to choose a state to process the task. If it switches from state $i$ to state $j$ and processes the task there, it incurs a cost $d(i,j)+c_j$. Given an initial state and a sequence of requests, the goal is to process all tasks at minimum cost.

It is easy to see that the online version of dynamic MSSC problem is a MTS, where the states correspond to permutations, thus $N = n!$, and the distance between two states is their Kendall tau distance. For a request set $S_t$, the request is a vector specifying the cost $\pi(S_t)$ for every permutation $\pi$. 

Several other fundamental online problems (e.g., $k$-server, convex body chasing) are MTS. 
Although there has been a lot of work on understanding the structure of MTS problems~\cite{BLS92,CL19,BCLL19,KP95,Sit14,Sel20,AGGT20,BEK17,BEKN18}, there is not a good grasp on how the structure relates to the hardness of MTS problems. Getting a better understanding on this area is a long-term goal, since it would lead to a systematic framework for solving online problems.


\subsection{Preliminaries}
\label{sec:prelim}

\paragraph{Notation.} Given a request sequence $\cS = \lbrace S_1, \dotsc, S_m \rbrace $, for any algorithm $\ALG$ we denote $\cost(\ALG(\cS))$ or simply $\cost(\ALG)$ the total cost of $\ALG$ on $\cS$. Similarly we denote $\acost(\ALG)$ the total access cost of $\ALG$ and $\mcost(\ALG)$ the total movement cost of $\ALG$. For a particular time step $t$, an algorithm using permutation $\pi_t$ incurs an access cost $\acost(\ALG(t)) = \pi_t(S_t)$. We denote by $\pi_t[j]$ the position of element $j \in U$ in the permutation $\pi_t$.

\paragraph{Online Min-Sum Set Cover.} We focus on the $r$-uniform case, i.e., when all sets $S_t$ have size $r \ll n$. This is essentially without loss of generality, because we can always let $r = \max_{t}|S_t|$ and add the $r - |S_t|$ last unrequested elements in the algorithm's permutation to any set $S_t$ with $|S_t| < r$. Assuming that $r \leq n/2$, this modification cannot increase the optimal cost and cannot decrease the online cost by more than a factor of $2$.

\end{onlymain}

\begin{onlymain}

\section{Lower Bounds on the Deterministic Competitive Ratio}
\label{sec:lower_bounds}

We start with a lower bound on the deterministic competitive ratio of online MSSC.

\smallskip
\noindent {\hf Theorem~\ref{thm:lb}.} {\em Any deterministic online algorithm for the Online Min-Sum Set Cover problem has competitive ratio at least $(r+1)(1 - \frac{r}{n+1})$. }
\smallskip

For the proof, we employ an averaging argument, similar to those in lower bounds for list update and $k$-server~\cite{MMS90,ST85}. In each step, the adversary requests the last $r$ elements in the algorithm's permutation. Hence, the algorithm's cost is at least $(n-r+1)$. Using a counting argument, we show that for any fixed set $S_t$ of size $r$ and any $i \in [n-r+1]$, the number of permutations $\pi$ with access cost $\pi(S_t) = i$ is $\binom{n-i}{r-1}r!(n-r)!$\,. Summing up over all permutations and dividing by $n!$, we get that the average access cost for $S_t$ is $\binom{n+1}{r+1}\frac{r!(n-r)!}{n!} = \frac{n+1}{r+1}$. Therefore, the cost of the optimal permutation is a most $\frac{(n+1)}{r+1}$, and the competitive ratio of the algorithm at least $\frac{(n-r+1)(r+1)}{n+1}$. The details can be found in Appendix~\ref{app:lower_bounds}. 

\paragraph{Lower Bounds for Generalizations of Move-to-Front.} %
For list update, where $r=1$, simple algorithms like Move-to-Front (MTF) and Frequency Count achieve an optimal competitive ratio. We next briefly describe several such generalizations of them and show that their competitive ratio depends on $n$, even for $r = 2$. Missing details can be found in Appendix~\ref{app:lower_bounds}.



\noindent
\textbf{$\text{MTF}_{\text{first}}$}: Move to the first position (of the algorithm's permutation) the element of $S_t$ appearing first in $\pi_{t}$\,. This algorithm is $\Omega(n)$-competitive when each request $S_t$ consists of the last two elements in $\pi_{t}$. Then, the last element in the algorithm's permutation never changes and is used by the optimal permutation to serve the entire sequence! 

\noindent
\textbf{$\text{MTF}_{\text{last}}$:} Move to the first position the element of $S_t$ appearing last in $\pi_{t}$\,.

\noindent
\textbf{$\text{MTF}_{\text{all}}$}: Move to the first $r$ positions all elements of $S_t$ (in the same order as in $\pi_{t}$). 

\noindent
\textbf{$\text{MTF}_{\text{random}}$}: Move to the first position an element of $S_t$ selected uniformly at random. 

$\text{MTF}_{\text{last}}$, $\text{MTF}_{\text{all}}$ and $\text{MTF}_{\text{random}}$ have a competitive ratio of $\Omega(n)$ when each request $S_t$ consists of a fixed element $e$ (always the same) and the last element in $\pi_{t}$, because they all incur an (expected for $\text{MTF}_{\text{random}}$) moving cost of $\Theta(n)$ per request. 

The algorithms seen so far fail for the opposite reasons: $\text{MTF}_{\text{first}}$ cares only about the first element and ignores completely the second, and the others are very aggressive on using the second ($r$th) element. A natural attempt to balance those two extremes is the following.

\noindent
\textbf{$\text{MTF}_{\text{relative}}$}: Let $i$ be the position of the first element of $S_t$ in $\pi_{t}$. Move to the first positions of the algorithm's permutation (keeping their relative order) all elements of $S_t$ appearing up to the position $c\cdot i$ in $\pi_{t}$, for some constant $c$. The bad instance for this algorithm is when each request $S_t$ consists of the last element and the element at position $\lfloor n/c \rfloor - 1$ in $\pi_{t}$; it never uses the $n$th element and the adversary serves all requests with it at a cost of 1. 

All generalizations of MTF above are memoryless and they all fail to identify the element by which optimal serves $S_t$. The following algorithm tries to circumvent this by keeping memory  and in particular the frequencies of reqested elements.

\noindent
\textbf{$\text{MTF}_{\text{count}}$}: Move to the first position the most frequent element of $S_t$ (i.e., the element of $S_t$ appearing in most requested sets so far). 



This algorithm behaves better in easy instances, however with some more work we can show a lower bound of $\Omega(\sqrt{n})$ on its competitive ratio. 
Let $e_1, \ldots, e_n$ be the elements indexed according to the initial permutation $\pi_0$ and $b = \sqrt{n}$. The request sequence proceeds in $m/n$ phases of length $n$ each. The first $n - b$ requests of each phase are $\{ e_1, e_2 \}, \{ e_1, e_3 \}, \ldots, \{ e_1, e_{n-b} \}$, and the last $b$ requests consist of $e_{n-b+i}$ and the element at position $n-b$ at the current algorithm's permutation, for $i = 1, \ldots, b$. An optimal solution can cover all the requests by the elements $e_1, e_{n-b+1}, \ldots, e_{n}$ with total cost $\Theta(m+n\sqrt{n})$. The elements $e_{n-b+1}, \ldots, e_{n}$ are never upgraded by $\text{MTF}_{\text{count}}$. Hence, the algorithm's cost is $\Theta(m \sqrt{n})$.

\end{onlymain}

\begin{onlyapp}
\section{Deferred Proofs of Section \ref{sec:lower_bounds}}\label{app:lower_bounds}

In this section, we include the proofs deferred from section~\ref{sec:lower_bounds}. First, we prove the general lower bound for the competitive ratio of any deterministic online algorithm for the Online Min Sum Set Cover problem.

\noindent {\bf Theorem~\ref{thm:lb}.} {\em Any deterministic online algorithm for the Online Min-Sum Set Cover problem has competitive ratio at least $(r+1) \cdot (1 - \frac{r}{n+1})$. }
\begin{proof}
Let $\ALG$ be any online algorithm. The adversary creates a request sequence in which every request is composed by the $r$ last elements of the current permutation of $\ALG$.
At each round $t$, $\ALG$ incurs an accessing cost of at least $n-(r-1)$. Thus for the whole request sequence of $m$ requests, $\cost(\ALG) \geq m \cdot (n-r+1)$.

The non-trivial part of the proof is to estimate the cost of the optimal static permutation. We will count total cost of all $n!$ static permutations and use the average cost as an upper bound on the optimal cost. For any request set $S_t$, 
we intend to find the total cost of the $n!$ permutations for $S_t$. To do this, we will count the number permutations that have access cost of $i$, for every $1 \leq i \leq n-(r-1)$. For such counting, there are two things to consider. First, in how many different ways we can choose the positions where the $r$ elements of $S_t$ are located and second how many different orderings on elements of $S_t$ and of $U \setminus S_t$ exist. We address those two separately.

 
 \begin{enumerate}[(i)]\itemsep.4em
 \item For a permutation $\pi$ that incurs an access cost of $i$, it follows that, from the elements in $S_t$, the first one in $\pi$ is located in position $i$ and no other element from the set is located in positions $j < i$. The other $r-1$ elements of $S_t$ are located among the last $n-i$ positions of $\pi$. There are $\binom{n-i}{r-1} $ different ways to choose the locations of those elements. 

 \item Once the positions of elements of $S_t$ have been fixed, there are $r!$ different ways to assign the elements in those positions, equal to the number of permutations on $r$ elements. Similarly, there are $(n-r)!$ different ways to assign elements of $U \setminus S_t$ to the $n-r$ remaining positions.
 \end{enumerate}

\noindent 
Gathering the above, we conclude that the number of permutations that incur access cost exactly $i$ for a fixed request $S_t$ is $$ \binom{n-i}{r-1}  r! (n-r)!.$$ 
The latter implies two basic facts:
\begin{enumerate}
    \item $\sum_{i=1}^{n-r+1}\binom{n-i}{r-1} r! (n-r)! = n!$ (since each permutation has a specific cost for request $S_t$).
    
    \item The total sum of access costs for fixed request of size $r$ is:
 \begingroup
 \allowdisplaybreaks
 \begin{align*}
 \text{Total-Access-Cost } &=  \sum_{i=1}^{n-r+1} i\cdot \binom{n-i}{r-1}  r! (n-r)!\\
 &= \sum_{i=1}^{n-r+1} \sum_{j = i}^{n-r+1} \binom{n-j}{r-1}r! (n-r)! & \text{By reordering the terms}\\
  &= \sum_{i=1}^{n-r+1} 
  \underbrace{\sum_{j = i}^{n-r+1} \binom{n-j}{r-1}r! (n-r)!}_{\text{permutations with access cost} \geq i }\\
  &= \sum_{i=1}^{n-r+1} \binom{n- i + 1}{r} r!(n-r)!\\
 &= r!(n-r)!\binom{n+1}{r+1}
 \end{align*}
 \endgroup
\end{enumerate}
\noindent where the last equality follows by the fact that $\sum_{i=1}^{n-r+1}\binom{n-i}{r-1} r! (n-r)! = n!$ (see number $1$ above) with $n\leftarrow n+1$ and $r\leftarrow r+1$. Hence for a request sequence of length $m$, we get that

$$
\cost(\OPT) \leq m \cdot \frac{r!(n-r)!}{n!}\binom{n+1}{r+1}
= m \cdot \frac{n+1}{r+1}.
$$
\noindent
We conclude that for any deterministic algorithm $\ALG$, we have:

$$\frac{\cost(\ALG)}{\cost(\OPT)} \geq  \frac{m \cdot (n-r+1)}{m \cdot \frac{n+1}{r+1}} =  (r+1) \cdot \bigg( 1 - \frac{r}{n+1}  \bigg) = r+1 - \frac{r(r+1)}{n+1}. \qedhere$$

\end{proof}
\medskip

\paragraph{Lower bounds for various algorithms.} Next, we prove the lower bounds for the competitive ratio of the several online algorithms generalizing the MTF algorithm. For all lower bounds we use request sets of size $r=2$. Recall that $\pi_{t}(j)$ the $j$th element of the permutation $\pi_t$ for $1\leq j \leq n$ 

\paragraph{$\text{MTF}_{\text{first}}$}: Move to the first position (of the algorithm's permutation) the element of $S_t$ appearing first in $\pi_{t}$\,.\\
\noindent   
\textbf{Lower bound}:
Let the request sequence $S_1,S_2, \ldots S_m$, in which
$S_t$ contains the last two elements of $\text{MTF}_{\text{first}}$' s permutation at round $t-1$. Formally, $S_t = \lbrace \pi_{t}(n-1), \pi_t(n) \rbrace$. $\text{MTF}_{\text{first}}$ moves the first element of the request in the first position and in every round the last element in $\text{MTF}_{\text{first}}$' s permutation remains the same ($\pi_{t}(n) = \pi_0(n)$). As a result, $\text{MTF}_{\text{first}}$ pays $\Omega(n)$ in each request, whereas $\OPT$ has the element $\pi_{0}[n]$ in the first position and just pays $1$ per request.

\paragraph{$\text{MTF}_{\text{last}}$}: Move to the first position the element of $S_t$ appearing last in $\pi_{t}$\,. \\
\noindent
\textbf{Lower bound}:
Let the request sequence $S_1,S_2, \ldots S_m$, in which each set $S_t$ always contains the last element of $\pi_t-1$ and the fixed element $1$. Clearly $\text{MTF}_{\text{last}}$ pays $\Omega(m \cdot n)$, while $cost(\OPT)=m$ by having element $1$ in the first place.

\paragraph{$\text{MTF}_{\text{all}}$}: Move to the first $r$ positions all elements of $S_t$ (in the same order as in $\pi_{t}$) .\\
\noindent
\textbf{Lower bound}: The same as previous.

\paragraph{$\text{MTF}_{\text{random}}$}: Move to the first position an element of $S_t$ selected uniformly at random. \\
\noindent
\textbf{Lower bound}: Let the request sequence $S_1,S_2, \ldots S_m$, in which each set $S_t$ always contains an element selected uniformly at random from $\pi_{t}$ and the fixed element $1$. Therefore elements in the last $n/2$ positions of $\pi_{t}$ have probability $1/2$ to be chosen.
At each round $t$, $\text{MTF}_{\text{random}}$ moves 
with probability $1/2$ to the first position of the list, the element of $S_t$ that was randomly selected. Thus at each round $t$, 
$\text{MTF}_{\text{random}}$ pays with probability $1/4$, moving cost at least $n/2$, meaning that the overall expected cost is at least $m\cdot n/8$. As a result, the ratio is $\Omega(n)$ since $\OPT$ pays $m$ by keeping element $1$ in the first position.

\paragraph{$\text{MTF}_{\text{relative}}$}: Let $i$ be the position of the first element of $S_t$ in $\pi_{t}$. Move to the first positions of the algorithm's permutation (keeping their relative order) all elements of $S_t$ appearing up to the position $c\cdot i$ in $\pi_{t}$, for some constant $c$. \\
\noindent
\textbf{Lower bound}: Let the request sequence $S_1, \ldots S_n$ in which $S_t$ contains the $\lfloor \frac{n-1}{c} \rfloor$th and the $n$th element of the list at round $t-1$. $\text{MTF}_{\text{relative}}$ never moves the last element and thus $\pi_n(0)$ belongs in all sets $S_t$. As in first case, this provides an $\Omega(n)$ ratio.

\paragraph{$\text{MTF}_{\text{count}}$}: Move to the first position the most frequent element of $S_t$ (i.e., the element of $S_t$ appearing in most requested sets so far).\\
\textbf{Lower bound}: The request sequence $S_1,\ldots,S_m$ is specifically constructed so that
$\text{MTF}_{\text{count}}$ never moves the last $b$ elements of the initial permutation $\pi_0$.
\begin{center}
    $\pi_0=[ \underbrace{x_1, \ldots, x_{n-b}}_{n-b  \text{ elements}}, \underbrace{x_{n-b+1}, \ldots ,x_n}_{b \text{ elements}} ]$
\end{center}
The constructed request sequence $S_1,\ldots,S_m$ will be composed by $m/n$ sequences of length $n$. Each piece of length $n$ will have the following form:
\medskip
\begin{enumerate}
     \item $n-b$ requests $\{ x_1,x_2 \}, \{x_1, x_3 \}, \ldots \{ x_1, x_{n-b} \}$ (all requests contain $x_1$).
    \item $\{
    \text{element in position } n-b , x_{n-b+i}\}$ for $i=1$ to $b$ (additional $b$ requests). 
\end{enumerate}
\medskip
After the requests of type 1, the list is the same as the initial one, since $x_1$ has frequency $n-b$ and $x_2, \ldots x_{n-b}$ have frequency $1$. Now consider the requests of type $2$. $\text{MTF}_{\text{count}}$ moves always to the front the element which is in position $n-b$, since has already been involved in a type $1$ request and has greater frequency. Therefore, $\text{MTF}_{\text{count}}$ pays $b\cdot(n-b)$. Repeating the same request sequence $m/n$ times, we can construct a sequence of length $m$. In this request sequence, $\OPT$ keeps the element $x_1$ in the first position and the elements $\{x_{n-b+1},\ldots,x_n\}$ in the next $b$ positions. Thus, $\OPT$ pays $(n-b)\cdot m/n$ for the requests of type~$1$ and 
$b^2\cdot m/n$ for the requests of type~$2$.
$\text{MTF}_{\text{count}}$ pays $(n-b)\cdot m/n$ for the requests of type~$1$ (same as $\OPT$), but 
$(n-b)\cdot b\cdot m/n$ for the requests of type~2.
Setting $b = \sqrt{n}$, we get a $\Omega(\sqrt{n})$ lower bound for the competitive ratio of $\text{MTF}_{\text{count}}$, which concludes this section.
\end{onlyapp}

\begin{onlymain}
\section{An Algorithm with Asymptotically Optimal Competitive Ratio}
\label{sec:static-ub}

Next, we present algorithm $\LRA$ (Algorithm~\ref{alg:lazy-det-MWU}) and analyze its competitive ratio. The following is the main result of this section: 

\medskip
\noindent  {\hf Theorem~\ref{thm:static_ub}.} {\em Deterministic online algorithm $\LRA$, presented in Algorithm~\ref{alg:lazy-det-MWU}, is $(5r+2)$-competitive for the static version of the Online Min-Sum Set Cover problem. }
\medskip

The remainder of this section is devoted to the proof of Theorem~\ref{thm:static_ub}. At a high-level, our approach is summarized by the following three steps:

\begin{enumerate}
\item We use as black-box the multiplicative weights update ($\MWU$) algorithm with learning rate $1/n^3$. Using standard results from learning theory, we show that its expected access cost is within a factor $5/4$ of $\OPT$, i.e.,
    $\acost(\MWU) \leq \frac{5}{4} \cost (\OPT)$
    (Section~\ref{sec:mwu}).
    
\item We develop an online rounding scheme, which turns any randomized algorithm $\mathcal{A}$ into a deterministic one, denoted $\deterministic(\cA)$, with access cost at most  $ 2r \cdot \expect [\acost(\mathcal{A})]$ (Section~\ref{sec:rounding}). However, our rounding scheme does not provide any immediate guarantee on the moving cost of $\deterministic(\cA)$.
    
    \item $\LRA$ is a lazy version of $\deterministic(\MWU)$ that updates its permutation only if $\MWU$'s distribution has changed a lot. A \textit{phase} corresponds to a time interval that $\LRA$ does not change its permutation. We show that during a phase:
   
    \begin{enumerate}[(i)]
        \item  The upper bound on the access cost increases, compared to $\deterministic(\MWU)$, by a factor of at most 2, i.e., $\acost(\LRA) \leq 4r \cdot \expect[\acost(\MWU)]$ (Lemma~\ref{lem:alg-acc_cost}).
        
        \item The (expected) access cost of $\MWU$ is at least $n^2$.  Since our algorithm moves only once per phase, its movement cost is at most $n^2$. Thus we get that (Lemma~\ref{lem:moving_cost}): $$\mcost(\LRA) \leq \expect[\acost(\MWU)] \, .$$
    \end{enumerate}
For the upper bound on the moving cost above, we relate how much $\MWU$'s distribution changes during a phase, in terms of the total variation distance, to the cost of $\MWU$ and the cost of our algorithm. 
\end{enumerate}

Based on the above properties, we compare the access and the moving cost of $\LRA$ against the access cost of $\MWU$ and to get the desired competitive ratio:
$$\cost(\LRA) \leq (4r + 1)\expect[\acost(\MWU)] \leq (5r + 2)\cost(\OPT) \, .$$


Throughout this section we denote by $d_{\tv}(\delta,\delta')$ the total variation distance of two discrete probability distributions $\delta, \delta': [N] \rightarrow [0,1] $, defined as $d_{\tv}(\delta,\delta') = \sum_{i=1}^{N} \max \lbrace 0, \delta(i) - \delta'(i) \rbrace$.

\subsection{Using Multiplicative Weights Update in Online Min-Sum Set Cover}
\label{sec:mwu}
In this section, we explain how the well-known $\MWU$ algorithm~\cite{LW94,FS97} is used in our context. 

\paragraph{The MWU Algorithm.} Given $n!$ permutations of elements of $U$, the algorithm has a parameter $\beta \in [0,1]$ and a weight $w_{\pi}$ for each permutation $\pi \in [n!]$, initialized at 1. At each time step the algorithm chooses a permutation according to distribution  $\mathrm{P}^t_\pi = w_\pi^t / (\sum_{\pi \in [n!]}w_\pi^t)$. When request $S_t$ arrives, $\MWU$ incurs an expected access cost of $$\expect[\acost(\MWU (t) )]= \sum_{\pi \in [n!]} \mathrm{P}_{\pi}^t \cdot \pi(S_t)$$
and updates its weights $w^{t+1}_\pi = w^{t}_\pi \cdot \beta^{\pi(S_t)}$, where $\beta = e^{-1/n^3}$; this is the so-called \textit{learning rate} of our algorithm. Later on, we discuss the reasons behind choosing this value.

\paragraph{On the Access Cost of MWU.} Using standard results from learning theory~\cite{LW94,FS97} and adapting them to our setting, we get that the (expected) access cost of $\MWU$ is bounded by $\cost(\OPT)$. This is formally stated in Lemma~\ref{lem:mwu-scaled-final} (and is proven in Appendix~\ref{sec:app-static-ub}).

\begin{lemma}
\label{lem:mwu-scaled-final}
For any request sequence $\sigma = (S_1,\ldots,S_m)$ we have that 
$$\expect[\acost(\MWU)] \leq \frac{5}{4} \cdot \cost(\OPT) + 2 n^4 \ln n \, .$$
\end{lemma}
\end{onlymain}

\begin{onlyapp}
\section{Deferred Proofs of Section~\ref{sec:static-ub}}
\label{sec:app-static-ub}
In this section we include the proofs deferred from Section~\ref{sec:static-ub}.

\subsection{Proofs Related to MWU Algorithm}
\label{app:mwu}
Here we include omitted proofs related to MWU algorithm. 

\paragraph{Access Cost of MWU.} We first show that the MWU is $5/4$-competitive for access costs. 

\noindent {\bf Lemma~\ref{lem:mwu-scaled-final}. } {\em For any request sequence $\sigma = (S_1,\ldots,S_m)$ we have that 
$$\expect[\acost(\MWU)] \leq \frac{5}{4} \cdot \cost(\OPT) + 2 n^4 \ln n.$$}
\begin{proof}

By the standard results in learning theory~\cite{LW94,FS97}, we know that for any sequence $\sigma = (S_1,\ldots,S_m)$, the MWU algorithm satisfies

\begin{equation*}
\label{eq:mwu-bound}
\sum_{t=1}^m \sum_{\pi \in [n!]}\mathrm{P}^t_\pi \cdot \acost (\pi,S_t)  \leq \frac{\ln(1/\beta)}{1-\beta} \cdot \min_{\pi \in [n!]}\sum_{t=1}^m \pi(S_t)+ \frac{\ln (n!)}{1 - \beta}.
\end{equation*}
where $\beta = e^{-1/n^3}$. Thus, $3/4 < \beta < 1$, for any $n\geq 2 $. Using standard inequalities we get that $\frac{\ln (1/\beta)}{1 - \beta} \leq 5/4 $ and $1-\beta \geq 1/2n^3$ for any $n \geq 2$. We finally get that,
$$
\expect[\acost(\MWU)] \leq \frac{5}{4} \cdot \cost(\OPT) + 2 n^4 \ln n.$$
\end{proof}
\end{onlyapp}

\begin{onlymain}
\paragraph{On the Distribution of MWU.} We now relate the expected access cost of the MWU algorithm to the total variation distance among $\MWU$'s distributions. More precisely, we show that if the total variation distance between $\MWU$'s distributions at times $t_1$ and $t_2$ is large, then $\MWU$ has incurred a sufficiently large access cost. The proof of the following makes a careful use of $\MWU$'s properties and is deferred to Appendix~\ref{sec:app-static-ub}.

\begin{lemma}
\label{lem:switch_cost}
Let $\mathrm{P}^t$ be the probability distribution of the MWU algorithm at time $t$. Then, 

$$    \text{d}_{\tv}(\mathrm{P}^t,\mathrm{P}^{t+1}) \leq \frac{1}{n^3} \cdot \expect[\acost(\MWU(t))]. $$
\end{lemma}
\end{onlymain}

\begin{onlyapp}
\paragraph{Total Variation Distance.} We now proceed on the proofs of lemmas relating the total variation distance of the distribution maintained by the MWU algorithm to its access cost.

{\bf Lemma~\ref{lem:switch_cost}. }{\em Let $\mathrm{P}^t$ be the probability distribution of MWU algorithm at time $t$. Then, the probability distribution $\mathrm{P}^{t+1}$ of the algorithm satisfies

$$    \text{d}_{\tv}(\mathrm{P}^t,\mathrm{P}^{t+1}) \leq \frac{1}{n^3} \cdot \expect[\acost(\MWU(t))]. $$
 }

\begin{proof}
To simplify notation, let $W^t = \sum_{\pi \in [n!]} w_\pi^t$. We remind that by the definition of $\MWU$, $w_{\pi}^{t+1} = w_{\pi}^{t} \cdot e^{-\pi(S_t)/n^3}$. Moreover, by the definition of total variation distance,

\begingroup
\allowdisplaybreaks
\begin{align*}
 d_{\tv}(\mathrm{P}^t,\mathrm{P}^{t+1}) 
 &= \sum_{\pi: \mathrm{P}_\pi^t > \mathrm{P}_\pi^{t+1}}\mathrm{P}_\pi^t - \mathrm{P}_\pi^{t+1} = \sum_{\pi: \mathrm{P}_\pi^t > \mathrm{P}_\pi^{t+1}} \bigg(  \frac{w_\pi^t}{W^t} - \frac{w_\pi^{t+1}}{W^{t+1}} \bigg)\\ 
 &\leq \sum_{\pi: \mathrm{P}_\pi^t > \mathrm{P}_\pi^{t+1}} \bigg(  \frac{w_\pi^t}{W^t} - \frac{w_\pi^{t+1}}{W^t} \bigg)   \\[15pt]
 &\leq \sum_{\pi \in [n!]} \bigg(  \frac{w_\pi^t}{W^t} - \frac{w_\pi^{t+1}}{W^t} \bigg)= \sum_{\pi \in [n!]} \frac{w_\pi^t}{W^t} \cdot \bigg( 1  - e^{-\pi(S_t)/n^3}  \bigg) \\[15pt]
 & =  \sum_{\pi \in [n!]} \mathrm{P}_\pi^t \cdot \bigg( 1  -  e^{-(\pi(S_t)/n^3)} \bigg) \leq \sum_{\pi \in [n!]} \mathrm{P}_\pi^t \cdot \frac{\pi(S_t)}{n^3}    \\[15pt]
 &= \frac{1}{n^3} \cdot \expect[\acost(\MWU(t))] \, .       
\end{align*}
\endgroup

\noindent In the first inequality we used that $W^{t+1} \leq W^t$. In the second inequality we used that for all $\pi$ we have that $w_\pi^{t+1} \leq w_\pi^t$ which implies that $\frac{w_\pi^t - w_\pi^{t+1}}{W^t} \geq 0$. In the last inequality we used that $1-e^{x} \leq -x$, for any $x$. \qedhere
\end{proof}
\end{onlyapp}

\begin{onlymain}
The following is useful for the analysis of $\LRA$. Its proof follows from Lemma~\ref{lem:switch_cost} and the the triangle inequality and is deferred to Appendix~\ref{sec:app-static-ub}.

\begin{lemma}
\label{lem:mwu-cost-dtv}
Let $t_1$ and $t_2$ two different time steps such that $d_{\tv}(\mathrm{P}^{t_1},\mathrm{P}^{t_2}) \geq 1/n$. Then, 
$$\sum_{t=t_1}^{t_2-1} \expect[\acost(\MWU(t))] \geq n^2 \, .$$
\end{lemma}
\end{onlymain}

\begin{onlyapp}

\noindent {\bf Lemma~\ref{lem:mwu-cost-dtv}. } {\em Let $t_1$ and $t_2$ two different time steps such that $d_{\tv}(\mathrm{P}^{t_1},\mathrm{P}^{t_2}) \geq 1/n$. Then, during the time interval $[t_1,t_2)$ the cost of the $\MWU$ algorithm is at least $n^2$.}
$$\sum_{t=t_1}^{t_2-1} \expect[\acost(\MWU(t))] \geq n^2.$$
\begin{proof}

By Lemma~\ref{lem:switch_cost} and summing over all $t$ such that $t_1 \leq t <t_2$, we have that 
\begin{equation}
\label{ref:eq_mwu-sum}
\sum_{t=t_1}^{t_2-1} d_{\tv}(\mathrm{P}^{t},\mathrm{P}^{t+1}) \leq \frac{1}{n^3} \cdot \sum_{t=t_1}^{t_2-1} \expect[\acost(\MWU(t))]. 
\end{equation}
By triangle inequality we have that $d_{\tv}(\mathrm{P}^{t_1},\mathrm{P}^{t_2}) \leq \sum_{t=t_1}^{t_2-1} d_{\tv}(\mathrm{P}^t,\mathrm{P}^{t+1}) $. Combined with~\eqref{ref:eq_mwu-sum}, this implies that
\[ d_{\tv}(\mathrm{P}^{t_1},\mathrm{P}^{t_2}) \leq \frac{1}{n^3} \sum_{t=t_1}^{t_2-1} \expect[\acost(\MWU(t))]. \]

By rearranging and using that $d_{\tv}(\mathrm{P}^{t_1},\mathrm{P}^{t_2}) \geq 1/n$, we get that 
\[ \sum_{t=t_1}^{t_2-1} \expect[\acost(\MWU(t))] \geq n^3 \cdot \frac{1}{n} = n^2 \qedhere \]
\end{proof}
\end{onlyapp}

\begin{onlymain}

\subsection{Rounding}
\label{sec:rounding}

Next, we present our rounding scheme. Given as input a probability distribution $\delta$ over permutations, it outputs a fixed permutation $\rho$ such that for each possible request set $S$ of size $r$, the cost of $\rho$ on $S$ is within a $O(r)$ factor of the expected cost of the distribution $\delta$ on $S$. For convenience, we assume that $n/r$ is an integer. Otherwise, we use $\lceil n/r \rceil$.

\begin{algorithm}[H]
  \caption{Greedy-Rounding  (derandomizing probability distributions over the permutations)}\label{alg:derand-sec}
  \textbf{Input:} A probability distribution $\delta$ over $[n!]$.\\
  \textbf{Output:} A permutation $\rho \in [n!]$.

 \begin{algorithmic}[1]
 
        \STATE R $\leftarrow U$
        \FOR{$i = 1$ \text{ to } $n/r$}
        
        \STATE $S^i \leftarrow \argmin_{S \in \{R\}^r} \expect_{\pi \sim \delta}[ \pi(S)]$
        
        \STATE Place the elements of $S^i$ (arbitrarily) from positions $(i-1)\cdot r + 1$ to $i \cdot r$ of $\rho$.
        \STATE $R \leftarrow R \setminus S^i$
        \ENDFOR
        
        \RETURN $\rho$
  \end{algorithmic}
\end{algorithm}

Our rounding algorithm is described in Algorithm~\ref{alg:derand-sec}. At each step, it finds the request $S$ with minimum expected covering cost under the probability distribution $\delta$ and places the elements of $S$ as close to the beginning of the permutation as possible. Then, it removes those elements from set $R$ and iterates. The main claim is that the resulting permutation has the following property:
\textit{any request $S$ of size $r$ has covering cost at most
$O(r)$ times of its expected covering cost under the probability distribution $\delta$}.

\begin{theorem}
\label{thm:greedy_rounding}
Let $\delta$ be a distribution over permutations and let $\rho$ be the permutation output by Algorithm~\ref{alg:derand-sec} on $\delta$. Then, for any set $S$, with $|S|=r$, $$\rho(S) \leq 2  r \cdot \expect_{\pi \sim \delta} [\pi( S)] \, .$$
\end{theorem}

\begin{proof}[Proof Sketch.]
The key step is to show that if the element used by $\rho$ to serve the request $S$ was picked during the $k$th iteration of the rounding algorithm, then $\expect_{\pi \sim \delta} [\pi( S)] \geq k/2$. Clearly, $\rho(S) \leq k \cdot r$ and the theorem follows. Full proof is in Appendix~\ref{app:rounding}.
\end{proof}

\end{onlymain}

\begin{onlyapp}

\subsection{Rounding}
\label{app:rounding}

\noindent {\bf Theorem~\ref{thm:greedy_rounding}.}{\em  Let $\delta$ be a distribution over permutations and let $\rho$ be the permutation output by Algorithm~\ref{alg:derand-sec} on $\delta$. Then, for any set $S$, with $|S|=r$, $$\rho(S) \leq 2  r \cdot \expect_{\pi \sim \delta} [\pi( S)] \, .$$}

\begin{proof}
Let $e$ be the element used by $\rho$ to serve the request on set $S$. Pick $k$ such that $(k-1) \cdot r +1 \leq \acost(\rho,S) \leq k \cdot r$. That means, $e$ was placed at the permutation $\rho$ at the $k$th iteration of the rounding algorithm. 

Let $S^1,\dotsc,S^{n/r}$ be the sets chosen during the rounding algorithm. Recall that $\rho$ uses an element from $S^k$ to serve the request. To this end, we use the technical Lemma~\ref{lem:rounding-opt-lower-bound} in order to get a lower bound on the expected cost of $\delta$. We distinguish between two cases: 

\begin{enumerate}
    \item Case $S = S^k$. In that case, by Lemma~\ref{lem:rounding-opt-lower-bound} we get that $\expect_{\pi \sim \delta} [\pi(S)] \geq \frac{k+1}{2}$.

    \item Case $S \neq S^k$. That means, $e$ is one element of $S$ in $S_k$ and no elements of $S$ are in sets $S_1,\dotsc,S_{k-1}$. By construction of or rounding algorithm, we have that$$\expect_{\pi \sim \delta} [\pi(S)] \geq \expect_{\pi \sim \delta}[\pi(S_k)] \geq \frac{k+1}{2}. $$
\end{enumerate}
We get that in both cases $\expect_{\pi \sim \delta} [\pi(S)] \geq \frac{k+1}{2}$. We conclude that

$$ \frac{\acost(\rho,S)}{\expect_{\pi \sim \delta}[\acost(\delta,S)]} \leq \frac{k \cdot r}{\frac{k+1}{2}} \leq 2 \cdot r \qedhere $$ 

\end{proof}

We now proceed to the lemma omitted in the proof of Theorem~\ref{thm:greedy_rounding}.

\begin{lemma}
\label{lem:rounding-opt-lower-bound}
Let $\delta$ be a probability distribution over permutations and $1 \leq k \leq \frac{n}{r}$. Let $S_1,\dotsc,S_{k}$ be disjoint sets such that $S_j \subseteq U$ and $|S_j| = r$ for all $1 \leq j \leq k$. Let $\cE_j = \expect_{\pi \sim \delta}[\pi(S_j)]$ for any $1\leq j \leq k$. If $\cE_1 \leq \dotsc \leq \cE_k$, then, we have that $\cE_j \geq \frac{j+1}{2}$, for $1 \leq j \leq k$.
\end{lemma}

\begin{proof}
We have that 

\begin{align*}
    \cE_j \geq \frac{1}{j} \sum_{\ell=1}^{j} \cE_{\ell} & = \frac{1}{j} \sum_{\ell=1}^{j} \sum_{\pi}  \Prob_{\delta}[\pi] \cdot \pi(S_{\ell}) & \text{Using $ \cE_1 \leq \dotsc \leq \cE_j$}   \\ &=  \frac{1}{j} \sum_{\pi \in [n!]} \Prob_{\delta}[\pi] \cdot \sum_{\ell=1}^{j} \pi(S_{\ell})  & \text{Linearity of summation}  \\
    & \geq \frac{1}{j} \sum_{\pi \in [n!]} \Prob_{\delta}[\pi] \cdot \frac{j(j+1)}{2} \\
    &= \frac{j+1}{2} \sum_{\pi} \Prob_{\delta}[\pi]    = \frac{j+1}{2}   
\end{align*}
where
$\sum_{\ell=1}^{j} \pi(S_{\ell}) \geq \frac{j(j+1)}{2}$ follows by the fact that $\pi(S_{\ell})$ take $j$ different positive integer values (the sets $S_\ell$ are disjoint).
\end{proof}
\end{onlyapp}

\begin{onlymain}
\subsection{The Lazy Rounding Algorithm}
\label{sec:alg}

$\LRA$, presented in Algorithm~\ref{alg:lazy-det-MWU}, is essentially a lazy derandomization of $\MWU$. At each step, it calculates the distribution on permutations maintained by $\MWU$. At the beginning of each \textit{phase}, it sets its permutation to that given by Algorithm~\ref{alg:derand-sec}. Then, it sticks to the same permutation for as long as the total variation distance of $\MWU$'s distribution at the beginning of the phase to the current $\MWU$ distribution is at most $1/n$. As soon as the total variation distance exceeds $1/n$, $\LRA$ starts a new phase. 


\begin{algorithm}[b!]
  \caption{Lazy Rounding}\label{alg:lazy-det-MWU}
  \textbf{Input:} Sequence of requests $(S_1,\ldots,S_m)$ and the initial permutation $\pi_0 \in [n!]$.\\
  \textbf{Output:} A permutation $\pi_t$ at each round $t$, which serves request $S_t$.

 \begin{algorithmic}[1]

 \STATE $\phase \leftarrow 0$
 \STATE $\mathrm{P^1}\leftarrow \text{uniform distribution over permutations}$
 \FOR{ each round $t \geq 1$ }
 
\IF{$\text{d}_{\text{tv}}(\mathrm{P}^{t},\mathrm{P}^{\phase}) \leq 1/n$}
        
        \STATE $\pi_t \leftarrow \pi_{t-1}$ 
        
        \ELSE
       \STATE $\pi_t \leftarrow \text{Greedy-Rounding}(\mathrm{P}^t)$
       \STATE $\phase \leftarrow t$
        \ENDIF
 
    \STATE Serve request $S_{t}$ using permutation $\pi_{t}$.
    \STATE  $w^{t+1}_\pi = w^{t}_\pi \cdot e^{-\pi(S_t)/n^3} $, for all permutations $\pi \in [n!]$.                               
    \STATE $\mathrm{P}^{t+1} \leftarrow$ Distribution on permutations of MWU, $ \mathrm{P}^{t+1}_\pi = w^{t}_\pi / (\sum_{\pi \in [n!]}w_\pi ^t)$.

  \ENDFOR
\end{algorithmic}
\end{algorithm}

The main intuition behind the design of our algorithm is the following. In Section~\ref{sec:rounding} we showed that Algorithm~\ref{alg:derand-sec} results in a deterministic algorithm with access cost no larger than $2r \expect[\acost(\MWU)]$. However, such an algorithm may incur an unbounded moving cost; even small changes in the distribution of $\MWU$ could lead to very different permutations after rounding. To deal with that, we update the permutation of $\LRA$ only if there are substantial changes in the distribution of $\MWU$. Intuitively, small changes in $\MWU$'s distribution should not affect much the access cost (this is formalized in Lemma~\ref{lem:dtv}). Moreover, $\LRA$ switches to a different permutation only if it is really required, which we use to bounds $\LRA$'s moving cost. 

\paragraph{Bounding the Access Cost.} 
We first show that the access cost of $\LRA$ is within a factor of $4r$ from the expected access cost of $\MWU$ (Lemma~\ref{lem:alg-acc_cost}). To this end, we first show that if the total variation distance between two distributions is small, then sampling from those distributions yields roughly the same expected access cost for any request $S$. The proof of the following is based on the optimal coupling lemma and can be found in Appendix~\ref{app:alg}.

\begin{lemma}
\label{lem:dtv}
Let $\delta$ and $\delta'$ be two probability distributions over permutations. If that $d_{\tv}(\delta,\delta') \leq 1/n$, for any request set $S$ of size $r$, we have that $$\expect_{\pi \sim \delta'} [\pi(S)] \leq 2 \cdot \expect_{\pi \sim \delta} [\pi(S)].$$
\end{lemma}

\end{onlymain}

\begin{onlyapp}

\subsection{Our Algorithm}
\label{app:alg}
In this section we include all proofs related to the analysis of our algorithm.

\noindent {\bf Lemma~\ref{lem:dtv}} {\em Let $\delta$ and $\delta'$ be two probability distributions over permutations. If that $d_{\tv}(\delta,\delta') \leq 1/n$, for any request set $S$ of size $r$, we have that $$\expect_{\pi \sim \delta'} [\pi(S)] \leq 2 \cdot \expect_{\pi \sim \delta} [\pi(S)].$$}

\begin{proof}
By coupling lemma there exists a coupling $(X,Y)$ such that 
\begin{equation}
    \label{eq:dtv}
    \Prob[X \neq Y] = d_{\tv}(\delta,\delta').
\end{equation}
Clearly, $\expect[\acost(X,S)] = \expect_{\pi \sim \delta} [\pi(S)]$ and $\expect[\acost(Y,S)] = \expect_{\pi \sim \delta'} [\pi(S)]$. Note that since minimum cost for serving a request is 1 and maximum $n$, it will always hold that $$1 \leq \expect[\acost(X,S)],\expect[\acost(Y,S)] \leq n.$$

We will show that $\expect[\acost(Y,S) - \acost(X,S)] \leq 1$. This implies the lemma as follows:
\begin{align*}
\expect_{\pi \sim \delta'} [\pi(S)]  &=   \expect[\acost(Y,S)]\\ &\leq \expect[\acost(X,S)] + 1 \leq 2 \expect[\acost(X,S)] \\&= 2 \expect_{\pi \sim \delta} [\pi(S)].
\end{align*}
Thus it remains to show that $\expect[\acost(Y,S) - \acost(X,S)] \leq 1$.
For notational convenience, let the random variable $Z = \acost(X,S) - \acost(Y,S)$. It suffices to show that $\expect[Z] \leq 1$. We have that
\begin{equation}
    \label{eq:dtv-expect-1}
    \expect[Z] \leq  \Prob[\acost(X,S) = \acost(Y,S)] \cdot 0 \\
    + \Prob[\acost(X,S) \neq \acost(Y,S)] \cdot n,  
\end{equation}
since whenever $X \neq Y$, the difference in the cost is upper bounded by $n$.

Observe that $\Prob[X(S) \neq Y(S)] \leq \Prob[X \neq Y]$; this is because if $X = Y$, then $\acost(X,S) = \acost(Y,S)$, but it may happen that $X \neq Y$ but $\acost(X,S) = \acost(Y,S)$. From~\eqref{eq:dtv-expect-1} we get that
\begin{equation}
\label{eq:dtv-expect-2}
\expect[Z] \leq \Prob[X \neq Y] \cdot n.    
\end{equation}
Combining~\eqref{eq:dtv} with~\eqref{eq:dtv-expect-2} we get that 
\[ \expect[Z] \leq d_{\tv}(\delta,\delta') \cdot n \leq \frac{1}{n} \cdot n = 1 . \qedhere \]
\end{proof}
\end{onlyapp}

\begin{onlymain}
We are now ready to upper bound the access cost of our algorithm.
\begin{lemma}
\label{lem:alg-acc_cost}
$\acost(\LRA) \leq  4 r \cdot \expect[\acost(\MWU)]$.
\end{lemma}

\begin{proof}
Consider a phase of $\LRA$ starting at time $t_1$. We have that at any round $t \geq t_1$,
$\pi_t = \text{Greedy-Rounding}(\mathrm{P}^{t_1})$, as long as $\text{d}_{\text{TV}}
(\mathrm{P}^t, \mathrm{P}^{t_1})
\leq 1/n$.
By Theorem~\ref{thm:greedy_rounding} and Lemma~\ref{lem:dtv}, we have that,
$$  \acost(\LRA(t)) = \pi_t(S_t) \leq 2r \cdot \expect_{\pi \sim \mathrm{P}^{t_1}}[\pi(S_t)] \leq 4r\expect_{\pi \sim \mathrm{P}^\text{t}}[\pi(S_t)].$$
 Overall we get,
$\acost(\LRA) = \sum_{t=1}^m \pi_t(S_t) \leq 4r \expect[\acost(\MWU)]. \qedhere$
\end{proof}

\paragraph{Bounding the Moving Cost.} 
We now show that the moving cost of $\LRA$ is upper bounded by the expected access cost of $\MWU$.

\begin{lemma}\label{lem:moving_cost}
$\mcost(\LRA) \leq \expect[\acost(\MWU)]$.
\end{lemma}

\begin{proof}
$\LRA$ moves at the end of a phase incurring a cost of at most $n^2$. Let $t_1$ and $t_2$ be the starting times of two consecutive phases. By the definition of $\LRA$, $d_{\tv}(P^{t_1},P^{t_2}) > 1/n$. By Lemma~\ref{lem:mwu-cost-dtv}, we have that the access cost of $\MWU$ during $t_1$ and $t_2$ is at least $n^2$. We get that
$$ \frac{\mcost(\ALG)}{\expect[\acost(\MWU)]} \leq \frac{n^2 \#\text{ different phases}}{n^2 \#\text{different phases}} = 1 .\qedhere  $$ 
\end{proof}

Theorem~\ref{thm:static_ub} follows from lemmas~\ref{lem:alg-acc_cost},~\ref{lem:moving_cost}~and~\ref{lem:mwu-scaled-final}. The details can be found in Appendix~\ref{app:alg}.

\end{onlymain}

\begin{onlyapp}

\paragraph{Proof of Theorem~\ref{thm:static_ub}.} We now give the formal proof of the competitive ratio of $\LRA$ algorithm.

\noindent  {\bf Theorem~\ref{thm:static_ub}.} {\em There exists a $(5r+2)$-competitive deterministic online algorithm for the static version of the Online Min-Sum Set Cover problem. }

\begin{proof}

We show that our algorithm $\ALG$ is $(5r+2)$-competitive. 

By lemmata~\ref{lem:alg-acc_cost} and~\ref{lem:moving_cost} we get that 
\begin{equation}
\label{eq:alg-mwu-compare}
\cost(\LRA) \leq (4r+1) \cdot \acost(\MWU). 
\end{equation}

Now, we connect the access cost of $\MWU$ to the optimal cost. By Lemma~\ref{lem:mwu-scaled-final} we have that 
\begin{equation}
\label{eq:mwu-opt-direct}    
\acost(\MWU) \leq \frac{5}{4} \cdot \cost(\OPT) + 2 n^4 \ln n.
\end{equation}

By~\eqref{eq:alg-mwu-compare} and~\eqref{eq:mwu-opt-direct} we get that 
\[  \cost(\LRA) \leq (5r+2) \cdot \cost(\OPT)+ 2\cdot (4r+1) \cdot n^4 \ln n \qedhere \]
\end{proof}

\end{onlyapp}

\begin{onlymain}

\paragraph{Remark.} Note that to a large extent, our approach is generic and can be used to provide static optimality for a wide range of online problems. The only requirement is that there is a maximum access cost $\costmax$ and a maximum moving cost $D$; then, we should use MWU with learning rate $1/(D \cdot \costmax)$ and move when $d_{\tv} \geq 1/\costmax$.  Here we used $D = n^2$ and $\costmax = n$. The only problem-specific part is the rounding of Section~\ref{sec:rounding}. We explain the details in Appendix~\ref{sec:app-static-ub}. We believe it is an interesting direction to use this technique for generalizations of this problem, like multiple intents re-ranking or interpret known algorithms for other problems like the BST problem using our approach.
\end{onlymain}

\begin{onlyapp}

\paragraph{On our Approach.} We now show how our approach for static optimality can be used to a wide class of online problems. 

Consider any MTS problem with $N$ states such that in each request the service cost of each state is at most $\costmax$ and the diameter of the state space is $D$. Assume that there exists a rounding scheme providing a derandomization of MWU such that the service cost is within $\alpha$ factor of the expected service cost of MWU. We explain how our technique from Section~\ref{sec:static-ub} can be used to obtain a $O(\alpha)$-competitive algorithm against the best state. 

\paragraph{Algorithm.} The algorithm is essentially the same as Algorithm~\ref{alg:lazy-det-MWU}, using the MWU algorithm with learning rate $1/(D \cdot \costmax)$ and moving when the total variation distance between the distributions exceeds $1/\costmax$.

This way, Lemma~\ref{lem:switch_cost} would give a bound of 
$$d_{\tv}(\mathrm{P}^t,\mathrm{P}^{t+1}) \leq \frac{1}{D} \cdot \acost(\MWU(t))$$
and as a consequence in Lemma~\ref{lem:mwu-cost-dtv} we will get that the (expected) cost of MWU during a phase is at least $D$.  Since the algorithm moves only at the end of the phase, incurring a cost of at most $D$, Lemma~\ref{lem:moving_cost} will still hold. Last, it is easy to see that Lemma~\ref{lem:dtv} will then hold for  $d_{\tv}(\delta,\delta') \leq 1/ \costmax$. This way, Lemma~\ref{lem:alg-acc_cost} would give that 
\[ \acost(\ALG) \leq 2 \alpha \expect[\cost(\MWU)]. \]

Combining the above, we get that the algorithm is $O(\alpha)$-competitive. 

\end{onlyapp}

\begin{onlymain}
\section{A Memoryless Algorithm}
\label{sec:algorithm}
In this section we focus on memoryless algorithms. We present an algorithm, called Move-All-Equally  (MAE), which seems to be the ``right'' memoryless algorithm for online MSSC. MAE decreases the index of all elements of the request  $S_t$ at the same speed until one of them reaches the first position of the permutation (see Algorithm~\ref{alg:move_all}). Note that $\MAE$ belongs to the \emph{Move-to-Front} family, i.e., it is a generalization of the classic MTF algorithm for the list update problem. $\MAE$ admits two key properties that substantially differentiate it from the other algorithms in the \text{Move-to-Front} family presented in Section~\ref{sec:lower_bounds}.

\begin{algorithm}[t]
  \caption{Move-All-Equally}\label{alg:move_all}
  \textbf{Input:} A request sequence $(S_1,\ldots,S_m)$ and the initial permutation $\pi_0 \in [n!]$\\
  \textbf{Output:} A permutation $\pi_t$ at each round $t$.

 \begin{algorithmic}[1]
  
  \FOR{each round $t \geq 1$}
  
  \STATE $k_t \leftarrow \min \{i \,|\, \pi_{t-1}[i] \in S_t  \} $ 
  
  \STATE Decrease the index of all elements of $S_t$ by $k_t-1$.
  
  \ENDFOR
  \end{algorithmic}
\end{algorithm}

\begin{enumerate}[(i)]
    \item Let $e_t$ denote the element used by $\OPT$ to cover the request $S_t$. $\MAE$ always moves the element $e_t$ towards the beginning of the permutation.
    \item It balances moving and access costs: if the access cost at time $t$ is $k_t$, then the moving cost of $\MAE$ is roughly $r \cdot k_t$ (see Algorithm~\ref{alg:move_all}). The basic idea is that the moving cost of $\MAE$ can be compensated by the decrease in the position of element $e_t$. This is why it is crucial all the elements to be moved with the same speed.
\end{enumerate}

\paragraph{Lower Bound.} First, we show that this algorithm, besides its nice properties, fails to achieve a tight bound for the online MSSC problem. 

\repeattheorem{thm:mae-lb}

In the lower bound instance, the adversary always requests the last $r$ elements of the algorithm's permutation. Since MAE moves all elements to the beginning of the permutation, we end up in a request sequence where $n/r$ disjoint sets are repeatedly requested. Thus the optimal solution incurs a cost of  $\Theta(n/r)$ per request, while MAE incurs a cost of $\Omega(n \cdot r)$ per request (the details are in  Appendix~\ref{app:mae}) . 
Note that in such a sequence,  MAE  loses a factor of $r$ by moving all elements, instead of one. However, this extra movement seems to be the reason that MAE outperforms all other memoryless algorithms and avoids poor performance in trivial instances, like other MTF-like algorithms. 

\end{onlymain}


\begin{onlyapp}

\section{Deferred Proofs of Section~\ref{sec:algorithm}}\label{app:mae}
In this section, we provide the proofs omitted from section~\ref{sec:algorithm}. More specifically, we prove a lower bound and an an upper bound on the competitive ratio of the $\MAE$ algorithm against the static optimal solution.

\subsection{Lower Bound}

 We start with the lower bound of $\Omega(r^2)$ on the competitive ratio of MAE against $\OPT$.
\repeattheorem{thm:mae-lb}
\begin{proof}
$\MAE$ is given an initial permutation $\pi_{0}$ with size $n$, where $n=r \cdot k$, for integers $r,k$ both greater than $1$. At each round $t$, the adversary gives requests $S_t$, which consist of the last $r$ elements in the permutation $\pi_{t-1}$ of $\MAE$. Since $\MAE$ moves all the elements of $S_t$ to the beginning of $\pi_t$, the request $S_{t+1}$ contains the $r$ elements preceding the elements of $S_t$ in $\pi_{t}$. Thus, the request sequence can be divided in $m/k$ requests containing the same $k$ pairwise disjoint requests, denoted as $S^*=\{ S^*_{1}, \ldots S^*_{k} \} $. The optimal static solution can serve the request sequence by using only $k$ elements, where each of these elements belongs to exactly one of the  $S^*_{j}, 1 \leq j \leq k$.  $\OPT$ has these elements in the first $k$ positions and pays $1+2+\ldots +k \leq k^2$ for every $k$ consecutive requests $(S^*_{1}, \ldots S^*_{k} ) $.   
$\MAE$ clearly pays $n-r+1$ access cost and $r \cdot(n-r+1)$ moving cost on every request, therefore $\MAE =m\cdot (r+1) \cdot(n-r+1)$. Then, the competitive ratio of $\MAE$ is at least \begin{eqnarray*}
\frac{\cost(\MAE)}{\cost(\OPT)} \geq \frac{m\cdot (r+1)\cdot (n-r+1)}{(m/k)\cdot k^2}=\frac{(r+1)\cdot (r\cdot k-r+1)}{k}=\Omega(r^2)\qedhere   
\end{eqnarray*}
\end{proof}
\end{onlyapp}

\begin{onlymain}

\paragraph{Upper Bounds.} Let $\cL$ denote the set of elements used by the optimal permutation on a request sequence such that $|\cL| = \ell$. That means, $\OPT$ has those $\ell$ elements in the beginning of its permutation, and it never uses the remaining $n-\ell$ elements. Consider a potential function $\Phi(t)$ being the number of inversions between elements of $\cL$ and $U \setminus \cL$ in the permutation of $\MAE$ (an inversion occurs when an element of $\cL$ is behind an element of $U \setminus \cL$). Consider the request $S_t$ at time $t$ and let $k_t$ be the access cost of $\MAE$. 

Let $e_t$ be the element used by $\OPT$ to serve $S_t$. Clearly, in the permutation of MAE, $e_t$ passes (i.e., changes relative order w.r.t) $k_t-1$ elements. Among them, let $L$ be the set of elements of $\cL$ and $R$ the elements of $U \setminus \cL$. Clearly, $|L| + |R| = k_t-1$ and $|L| \leq |\cL| = \ell$. We get that the move of $e_t$ changes the potential by $-|R|$. The moves of all other elements increase the potential by at most $(r-1) \cdot \ell$. We get that 
\begin{eqnarray*}
k_t + \Phi(t) - \Phi(t-1) &\leq&  |L| + |R| - |R| + (r-1) \cdot \ell 
                          \leq |L| +  (r-1) \cdot \ell \leq r \cdot \ell. 
\end{eqnarray*}
Since the cost of MAE at step $t$ is no more than $(r+1) \cdot k_t$, we get that the amortized cost of $\MAE$ per request is $O(r^2 \cdot \ell)$. This implies that for all sequences such that $\OPT$ uses all elements of $\cL$ with same frequencies (i.e, the $\OPT$ pays on average $\Omega(\ell)$ per request), $\MAE$ incurs a cost within $O(r^2)$ factor from the optimal. Recall that all other MTF-like algorithms are $\Omega(\sqrt{n})$ competitive even in instances where $\OPT$ uses only one element! 

While this simple potential gives evidence that $\MAE$ is $O(r^2)$-competitive, it is not enough to provide satisfactory competitiveness guarantees. We generalize this approach and define the potential function $ \Phi(t) = \sum_{j=1}^n \alpha_j \cdot \pi_t(j)$, where $\pi_t(j)$ is the position of element $j$ at round $t$ and $\alpha_j$ are some non-negative coefficients. The potential we described before is the special case where $\alpha_j = 1$ for all elements of $\cL$ and $\alpha_j = 0$ for elements of $U \setminus \cL$. 


By refining further this approach and choosing coefficients $\alpha_j$ according to the frequency that $\OPT$ uses element $j$ to serve requests (elements of high frequency are ``more important'' so they should have higher values $\alpha_j$), we get an improved upper bound.

\begin{reptheorem}{t:move_all_static}
The competitive ratio of $\MAE$ algorithm is at most $2^{O(\sqrt{\log n \cdot \log r})}$.
\end{reptheorem}

Note that this guarantee is $o(n^\epsilon)$ and $\omega(\log n)$. The proof is based on the ideas sketched above but the analysis is quite involved and is deferred to Appendix~\ref{app:mae}.

\end{onlymain}

\begin{onlyapp}

\subsection{Upper Bound}

 We now proceed to the proof of the upper bound on the competitive ratio of MAE algorithm. 

\repeattheorem{t:move_all_static}
Let the frequency of an element $e$ be the fraction of requests served by $e$ in the optimal permutation.
At a high-level, we divide the optimal permutation into $k+1$ \textit{blocks} where in each of the first $k$ blocks the frequencies of all elements are within a factor of $\beta$ and the last block contains all elements with frequencies at most $1/n^2$. By construction, in worst-case $k \approx \frac{\log n}{\log \beta}$. 

The main technical contribution of this section is proving the following lemma. 

\begin{lemma}\label{l:main}
For any parameter $\beta > 1$,
$$\cost(\MAE)\leq 4\beta(2r)^{k+1} \cdot \cost(\OPT) + (2r)^{k}n^2$$
where $k = \lceil 2 \log n / \log \beta \rceil$.
\end{lemma}

Note that Theorem~\ref{t:move_all_static} follows from Lemma~\ref{l:main} by balancing the values of $\beta$ and $(2r)^{k+1}$. It is easy to verify that by setting $\beta = 2^{\sqrt{8 \log n \cdot  \log 2r }}$, we obtain a competitive ratio at most $ \beta^2  \leq 2^{8 \sqrt{\log n \log r}} $







\paragraph{Notation and Definitions.} Let $f_j\in [0,1]$ denote the \emph{covering frequency}, which is the total fraction of requests served by the optimal permutation with element $j$. For convenience, we reorder the elements such that $f_1 \geq f_2 \geq \ldots \geq f_n$. As a result,  $\cost(\OPT) = m\cdot \sum_{j=1}^n j \cdot f_j$.

We partition the elements of $U$ into $k+1$ \textit{blocks}, as follows. The last block, $E_{k+1}$ contains all elements $j$ with $f_j \leq 1/n^2$; intuitively those are the least important elements. The block $E_i$ for $1 \leq i\leq k$ contains all elements with frequencies in the interval $[f_1/\beta^{i-1},f_1/\beta^{i})$. Note that in worst case, $k = \lceil 2 \log n / \log \beta \rceil$; we need that $f_1/\beta^{k} \leq 1/n^2$, which is equivalent to $k \geq f_1 \cdot \frac{2 \log n}{\log \beta}  $. 


    
    



\paragraph{Lower Bound on optimal cost.} Using this structure, we can express a neat lower bound on the cost of the optimal static permutation, which is formally stated in the following lemma. Lemma~\ref{l:2} provides a lower bound on $\cost(\OPT)$ using  the first $k$ blocks. For the rest of this section, let $\cF_i = \cup_{ j = 1}^i E_j $ the set of all elements in the first $i$ blocks, for $1\leq i \leq k$.  Let also $f_{\max}^i ,f_{\min}^i$ denote the maximum and minimum covering frequencies in the set $E_i$.

\begin{lemma}\label{l:2}
For any request sequence the optimal cost is at least
$$\cost(\OPT) \geq \frac{m}{2} \sum_{i=1}^k |E_i|\cdot |\cF_i| \cdot f_{\min}^i .$$
\end{lemma}

\begin{proof}
Using the definitions above, $\OPT$ uses element $j$ exactly $ m \cdot f_j$ times for a total cost $ j\cdot m  \cdot f_j$. To account for all elements of the same bucket $E_i$ together, we underestimate this cost and for an element $j \in E_i$ we charge $\OPT$ only for $m \cdot f_{\min}^i \leq m \cdot f_j$ requests. We get that:
\begin{eqnarray*}
\cost(\OPT) &=& m \sum_{j=1}^n f_j \cdot j \geq m \sum_{i=1}^k \sum_{j \in E_i}f_j \cdot j \geq m \sum_{i=1}^k f_{\min}^i \sum_{j \in E_i}j
 \\
&=& m \sum_{i=1}^k f_{\min}^i \sum_{j =1}^{|E_i|}(j + |\cF_{i-1}|)= m \sum_{i=1}^k f_{\min}^i \left(|E_i||\cF_{i-1}| + \frac{|E_i|(|E_i| + 1)}{2} \right )\\
&\geq& \frac{m}{2} \sum_{i=1}^k f_{\min}^i \cdot |E_i|\cdot(|\cF_{i-1}| + |E_i|) = \frac{m}{2} \sum_{i=1}^k f_{\min}^i \cdot |E_i|\cdot |\cF_{i}| \qedhere
\end{eqnarray*}
\end{proof}

\paragraph{Upper Bound on cost of $\MAE$.} Using the blocks, we also obtain an upper bound on the cost of Move-All-Equally. This is formally stated in Lemma~\ref{l:3}.

\begin{lemma}\label{l:3}
For any request sequence the cost of MAE algorithm can be upper bounded as follows:
$$\cost(\MAE) \leq (2r)^{k+1} \cdot m  \sum_{i=1}^k  |\cF_{i}| |E_i| \cdot f_{\max}^i  + 2(2r)^{k+1}\cdot m + (2r)^{k}\cdot n^2.$$
\end{lemma}
\noindent Before proceeding to the proof of Lemma~\ref{l:3}, which is quite technically involved, we show how Lemma~\ref{l:main} follows from Lemmata~\ref{l:3} and~\ref{l:2}.


\begin{proof}[Proof of Lemma~\ref{l:main}]
From Lemma~\ref{l:3} we have that
\begin{align*}
\cost(\MAE) &\leq (2r)^{k+1} \cdot \underbrace{m  \sum_{i=1}^k  |\cF_{i}| |E_i| \cdot f_{\max}^i}_{\leq \beta \cdot \cost(\OPT)}  +  \underbrace{ \vphantom{\sum_{i=1}^k  |\cF_{i}| |E_i| } 2(2r)^{k+1} \cdot m}_{\leq 2(2r)^{k+1} \cdot \cost(\OPT)} + (2r)^{k}\cdot n^2\\
 & \leq 4\beta (2r)^{k+1} \cdot \cost(\OPT) + (2r)^{k}\cdot n^2,
\end{align*}
where in the inequality we used the lower bound on $\cost(\OPT)$ from Lemma~\ref{l:2} and that $f_{\max}^i = f_{\min}^i \cdot \beta$ for all $i$.
\end{proof}

\subsection*{Proof of Lemma~\ref{l:3}}\label{s:pl3}
It remains to prove Lemma~\ref{l:3} which gives the upper bound on the cost of $\MAE$ algorithm. 
The proof lies on the right selections of the coefficients $\alpha_j$ in the potential function $\Phi(t)$:

$$ \Phi(t) = \sum_{j=1}^n \alpha_j \cdot \pi_t(j),$$
where $\pi_t(j)$ is the position of element $j$ at round $t$ and $\alpha_j$ are some non-negative coefficients. More precisely, if $j \in E_i$ then $\alpha_j  = (2r)^{k - i}$ for $i =1 , \ldots , k$ and $\alpha_j  = 0$ if $j \in E_{k+1}$.

At time $t$, let $k_t$ denote the access cost of $\MAE$ and let $e_t$ be the element used by $\OPT$ to serve request $S_t$. Using the coefficients mentioned above, we can break the analysis into two different types of requests: (i) requests served by $\OPT$ using an element $e_t \in E_i$ for $1 \leq i \leq k$ and (ii) requests served by $\OPT$ using $e_t \in E_{k+1}$. At a high-level the first case is the one where the choice of coefficients is crucial; for the second, we show that the frequencies are so ``small'', such that even if $\MAE$ incurs an access cost of $n$, this does not affect the bound of Theorem~\ref{t:move_all_static}.

We start with the first type of requests. We show the following lemma. 

\begin{lemma}\label{l:4}
At time $t$, if $e_t \in E_i$, for $1 \leq i \leq k$, then we have that 
$$\acost(\MAE(t))+ \Phi(t) - \Phi(t-1) \leq (2r)^k \cdot |\cF_{i} |$$
\end{lemma}

For the second type of requests we have the following lemma.
\begin{lemma}\label{l:5}
The total amortized cost of $\MAE$ algorithm for all requests such that $e_t \in E_{k+1}$ is at most $\lrp{(2r)^k+1} \cdot m$.
\end{lemma}

The proofs of those Lemmas are at the end of this section. We now continue the proof of Lemma~\ref{l:3}.

\begin{proof}[Proof of Lemma~\ref{l:3}]
We upper bound the total access cost of $\MAE$, that is $\sum_{t=1}^m k_t$. Clearly, by construction of $\MAE$, the total cost is at most $(r+1) \cdot \sum_{t=1}^m k_t $. 
We have that
\begingroup
\allowdisplaybreaks
\begin{align*}
\sum_{t = 1}^m k_t + \Phi(t) - \Phi(t-1) =&
\sum_{i=1}^k \lrp{ \sum_{t: e_t \in E_i}k_t + \Phi(t) - \Phi(t-1)
+ \sum_{t: e_t \in E_{k+1}}k_t + \Phi(t) - \Phi(t-1)}\\
& \leq \sum_{i=1}^k \sum_{t: e_t \in E_i}(2r)^k \cdot |\cF_{i} | +  \lrp{(2r)^k+1} \cdot m\\
& = \sum_{i=1}^k (2r)^k \cdot |\cF_{i}| \cdot m  \cdot \sum_{j \in E_i} f_j^t  +  \lrp{(2r)^k+1} \cdot m\\
& \leq \sum_{i=1}^k (2r)^k \cdot |\cF_{i}| \cdot m  \cdot \sum_{j \in E_i} f_{\max}^i   +  \lrp{(2r)^k+1} \cdot m\\
& = \sum_{i=1}^k (2r)^k \cdot |\cF_{i}| |E_{i}| \cdot f_{\max}^i \cdot m  +  \lrp{(2r)^k+1} \cdot m,
\end{align*}
\endgroup
where the first inequality comes from Lemmata~\ref{l:4},~\ref{l:5}. Using $\Phi(0) \leq (2r)^{k-1}n^2$ we get
\begingroup
\allowdisplaybreaks
\begin{align*}
\cost(\MAE) & \leq 2r \cdot \sum_{t=1}^m k_t \leq 2r \lrp{ \sum_{i=1}^k (2r)^k \cdot |\cF_{i}| |E_{i}| \cdot f_{\max}^i \cdot m + \lrp{(2r)^k + 1} \cdot m +  (2r)^{k-1}n^2} \\
   &\leq (2r)^{k+1}\sum_{i=1}^k  \cdot |\cF_{i}| |E_{i}| \cdot f_{\max}^i \cdot m + 2(2r)^{k+1}\cdot m + (2r)^{k}\cdot n^2 \qedhere
\end{align*}
\endgroup
\end{proof}

We conclude the Section with the proofs of lemmata~\ref{l:4},~\ref{l:5}, which were omitted earlier.

\begin{proof}[Proof of Lemma~\ref{l:4}]
At time $t$, the access cost of $\MAE$ is $k_t$. We have that
\begin{eqnarray*}
k_t + \Phi(t) - \Phi(t-1) 
&=& k_t - \alpha_{e_t} \cdot k_t + \sum_{j \neq e_t} \alpha_j ( \pi_t(j) -\pi_{t-1}(j) ) \\
&\leq& k_t \cdot (1 - \alpha_{e_t} ) + (2r)^{k-1} \cdot |\cF_{i}| \cdot r + \sum_{j \in U\setminus \cF_{i}} \alpha_j ( \pi_t(j) - \pi_{t-1}(j) ),
\end{eqnarray*}
where the inequality follows from the fact that each element in $\cF_{i}$ can increase its position by at most $r$ and that the maximum coefficient $\alpha_j$ is at most $(2r)^{k-1}$. We complete the proof by showing that 
$$ k_t \cdot (1 - \alpha_{e_t} ) + \sum_{j \in U\setminus \cF_{i}} \alpha_j ( \pi_t(j) - \pi_{t-1}(j) )
\leq 0 $$
If $e_t \in E_k$ then $\alpha_{e_t}=1$ and $\alpha_j = 0$ for all $j \in U \setminus \cF_k$, thus the left hand side equals 0 and the inequality holds. It remains to analyze the case where $e_t \in E_i$ and $i < k$.
\begin{align*}
k_t \cdot (1 - \alpha_{e_t} ) + \sum_{j \in U\setminus \cF_{i}} \alpha_j ( \pi_t(j) - \pi_{t-1}(j) ) &\leq  k_t \cdot (1 - \alpha_{e_t} ) + 
\sum_{j \in U\setminus \cF_{i}} \frac{\alpha_{e_t}}{2r} \cdot |\pi_t(j) - \pi_{t-1}(j) |\\
&\leq k_t \cdot (1 - \alpha_{e_t} )   + \frac{\alpha_{e_t}}{2r} \cdot r \cdot k_t \\
&= k_t \cdot( 1 -\alpha_{e_t} + \frac{\alpha_{e_t}}{2}) \leq 0.
\end{align*}
\noindent The last inequality is due to the fact that $\alpha_{e_t} \geq r$, thus $1- \alpha_{e_t}/2$ is negative if $r \geq 2$.
\end{proof}

\begin{proof}[Proof of Lemma~\ref{l:5}]
We sum the total amortized cost over all time steps $t$ such that $e_t \in E_{k+1}$. We have that
\begin{align*}
\sum_{t: e_t \in E_{k+1}}k_t + \Phi(t) - \Phi(t-1) &\leq \sum_{t: e_t \in E_{k+1}} n + (2r)^{k-1} \cdot n \cdot r \leq \lrp{(2r)^k + 1} \cdot n  \sum_{t: e_t \in E_{k+1}} 1\\
&\leq \lrp{(2r)^k + 1} \cdot n \cdot m \cdot
\sum_{ j \in E_{k+1}}f_j \leq \lrp{(2r)^k + 1} \cdot n \cdot m \cdot
\frac{|E_{k+1}|}{n^2}\\
&\leq \lrp{(2r)^k + 1} \cdot m \qedhere
\end{align*}
\end{proof}

\end{onlyapp}

\begin{onlymain}

\section{Dynamic Online Min-Sum Set Cover }
\label{sec:dynamic}

In this section, we turn our attention to the dynamic version of  online MSSC. 
In online dynamic MSSC, the optimal solution maintains a trajectory of permutations $\pi_0^\ast, \pi_1^\ast, \ldots, \pi_t^\ast, \ldots$ and use permutation $\pi^\ast_t$ to serve each request $S_t$. The cost of the optimal dynamic solution is $\OPT_\text{dynamic} = \sum_{t} (\pi^\ast_t(S_t) + d_{\text{KT}}(\pi^\ast_{t-1}, \pi^\ast_t))$, where $\{ \pi^\ast_t \}_{t}$ denotes the  optimal permutation trajectory for the request sequence that minimizes the total access and moving cost.

We remark that the ratio between the optimal static solution and the optimal dynamic solution can be as high as $\Omega(n)$. For example, in the sequence of requests $\{1\}^{b}\{2\}^{b}\ldots \{n\}^{b}$, the optimal static solution pays $\Theta(n^2 b)$, whereas the optimal dynamic solution pays $\Theta(n^2 + n \cdot b)$ by moving the element that covers the next $n \cdot b$ requests to the first position and then incurring access cost $1$. The above example also reveals that although Algorithm~\ref{alg:lazy-det-MWU} is $\Theta(r)$-competitive against the optimal static solution, its worst-case ratio against a dynamic solution can be $\Omega(n)$.

\paragraph{MAE Algorithm.} As a first study of the dynamic problem, we investigate the competitive ratio of \textit{Move-All-Equally} (MAE) algorithm from Section~\ref{sec:algorithm}. We begin with an upper bound.

\repeattheorem{thm:mae-dyn-ub}

\noindent The approach for proving Theorem~\ref{thm:mae-dyn-ub} is generalizing that exhibited in Section~\ref{sec:algorithm} for the static case. We use a generalized potential function $\Phi(t) = \sum_{j =1}^n \alpha_j^t \cdot \pi_t(j)$; i.e, the multipliers $\alpha_j$ may change over time so as to capture the moves of $\OPT_\text{dynamic}$. To select coefficients $\alpha_j^t$ we apply a two-level approach. We observe that there is always a 2-approximate optimal solution that moves an element of $S_t$ to the front (similar to classic MTF in list update). We call this $\MTF_{\OPT}$\,. We compare the permutation of the online algorithm with the permutation maintained by this algorithm; at each time, elements the beginning of the offline permutation are considered to be ``important'' and have higher coefficients $\alpha_j^t$. The formal proof is in Appendix~\ref{app:dynamic}.

Next, we show an almost matching lower bound.

\repeattheorem{thm:mae-dyn-lb}

\paragraph{Sketch of the Construction.}
 The lower bound is based on a complicated adversarial request sequence; we sketch the main ideas. Let $k$ be an integer. During a \textit{phase} we ensure that:

\begin{enumerate}[(i)]
    \item  There are $2k$ ``important'' elements used by $\OPT$; we call them $e_1,\dotsc,e_{2k}$. In the beginning of the phase, those elements are ordered in the start of the optimal permutation $\pi^\ast$, i.e., $\pi^{\ast}[e_j] = j$. The phase contains $k$ consecutive requests to each of them, in order; thus the total number of requests is $\approx 2k^2$. $\OPT$ brings each element $e_j$ at the front and uses it for $k$ consecutive requests; thus the access cost of $\OPT$ is $2k^2$ (1 per request) and the total movement cost of $\OPT$ of order $\Theta(k^2)$. Over a phase of $2k^2$ requests, $\OPT$ incurs an overall cost $\Theta(k^2)$, i.e., an average of $O(1)$ per request.
    \item The first $k+r-2$ positions of the online permutation will be always occupied by the same set of ``not important'' elements; at each step the $r-2$ last of them will be part of the request set and MAE will move them to the front. Thus the access cost will always be  $k+1$ and the total cost more than $(r+1) \cdot k$.
\end{enumerate}
The two properties above are enough to provide a lower bound $\Omega(r \cdot k)$; the optimal cost is $O(1)$ per request and the online cost $\Omega(r \cdot k)$. The goal of an adversary is to construct a request sequence with those two properties for the largest value of $k$ possible. 

The surprising part is that although MAE moves all requested elements towards the beginning of the permutation, it never manages to bring any of the ``important'' elements in a position smaller than $r+k-2$. While the full instance is complex and described in Appendix~\ref{app:dynamic}, at a high-level, we make sure that whenever a subsequence of $k$ consecutive requests including element $e_j$ begins, $e_j$ is at the end of the online permutation, i.e., $\pi_t[e_j] = n$. Thus, even after $k$ consecutive requests where $\MAE$ moves it forward by distance $k$, it moves by $k^2$ positions; by making sure that $n-k^2> r+k-2$ (which is true for some $k = \Omega(\sqrt{n})$), we can make sure that $e_j$ does not reach the first $r+k-2$ positions of the online permutation.

\end{onlymain}

\begin{onlyapp}

\section{Deferred Proofs of Section~\ref{sec:dynamic}}
\label{app:dynamic}
In this section we present the full proofs of Theorems~\ref{thm:mae-dyn-ub} and~\ref{thm:mae-dyn-lb}. 

\paragraph{Upper bound.} We start with the upper bound on the competitive ratio of the $\MAE$ algorithm. Recall that we use the generalized potential function $\Phi(t) = \sum_{j =1}^n \alpha_j^t \cdot \pi_t(j)$ for the dynamic case, where  the multipliers $\alpha_j$ change over time so as to capture the moves of $\OPT_\text{dynamic}$. 

\repeattheorem{thm:mae-dyn-ub}

Given the solution of $\OPT_\text{dynamic}$, we construct the solution $\textit{Move-To-Front}_{ \, \text{OPT}}$ $(\text{MTF}_{ \OPT})$, that at each round $t$ moves $e_t$ (the element $\OPT_\text{dynamic}$ uses to cover $S_t$) in the first position of the permutation. We compare the cost of $\MAE$, with the cost of $\text{MTF}_{\OPT}$, which is just a $2$-approximation of $\OPT_\text{dynamic}$ (for the same reason than MTF is 2-competitive for list update -- see Lemma~\ref{l:move_to_front} for a proof). More precisely, $\alpha_j^t=1$ if the element $j$ is in one of the first $c$ positions of $\text{MTF}_{\OPT}$'s permutation at round $t$ and $\alpha_j^t=0$ otherwise (the exact value of $c$ will be determined later).

\noindent Recall that the access cost of $\MAE$ at round $t$ is $k_t$. We have that at each round $t \geq 1$,
\begin{eqnarray}\label{eq:technical}
k_t + \Phi(t) - \Phi(t-1) &=& k_t + \sum_{j=1}^n \alpha_j^t \cdot \pi_t(j) -
\sum_{j=1}^n \alpha_j^{t-1} \cdot \pi_{t-1}(j)\nonumber\\
&=& (1 - \alpha_{e_t}^t) \cdot k_t + 
\sum_{j \neq e_t}\alpha_j^t \cdot (\pi_{t}(j) - \pi_{t-1}(j)) \nonumber\\
&+& \sum_{j =1}^n(\alpha_j^t - \alpha_j^{t-1})\cdot \pi_{t-1}(j) \nonumber\\
&\leq&r \cdot c + \underbrace{\left(\sum_{j=1}^n \alpha_j^t - \alpha_j^{t-1} \right )}_{\text{bounded by } \text{MTF}_{\text{OPT}}(t)} \cdot n,
\end{eqnarray}
\noindent where in the last inequality we used that $\alpha_{e_t}^t=1$ and that for each element, its position change from round $t-1$ to round $t$ $\pi_{t}(j) - \pi_{t-1}(j) $ can be at most $r$.


\begin{lemma}\label{l:technical} The term $\sum_{j=1}^n(\alpha_j^t - \alpha_j^{t-1})$, which appears in the difference $\Phi(t)-\Phi(t-1)$ of the potential function $\Phi(t)$ can be bounded by the moving cost of the  $\MTF_{\OPT}(t)$ as follows:
$\sum_{j=1}^n(\alpha_j^t - \alpha_j^{t-1})
\leq \alpha_{e_t}^t - \alpha_{e_t}^{t-1}  \leq 
 \mcost(\MTF_{\OPT}(t))/c.
$
\end{lemma}
\begin{proof}
For the first part of the inequality, observe that from round $t-1$ to round $t$, $\MTF_{\OPT}$ only moves element $e_t$ towards the beginning of the permutation. The latter implies, that for all elements $j \neq e_t$, $\alpha_j^t \leq \alpha_j^{t-1}$. For the second part of the inequality, observe that if $\alpha_{e_t}^{t-1} =0 $, then by the definition of the coefficients, $e_t$ does not belong in the first $c$ positions (of $\MTF_{\OPT}$'s permutation) at round $t-1$. Since at round $t$, $\MTF_{\OPT}$ moves $e_t$ in the first position of the list, $\mcost(\MTF_{\OPT}(t)) \geq c$
\end{proof}

Notice that the overall access cost of both $\MAE$ and $\text{MTF}_{\OPT}$ is just $m$ since at each round $t$ both of the algorithms admit an element, covering $S_t$, in the first position of the permutation. As a result by setting $c = \sqrt{n/r}$ and applying Lemma~\ref{l:technical},

\begin{eqnarray*}
\cost(\text{MAE}) &=& \acost(\MAE)+\mcost(\MAE)\\
&=& m+ r \sum_{t=1}^mk_t  \\
&\leq& m+   2r^{3/2}\sqrt{n}\sum_{t=1}^m\mcost(\text{MTF}_{\OPT}(t)) + \Phi(0) - \Phi(m)\\
&\leq& 2r^{3/2}\sqrt{n} \cdot \cost(\text{MTF}_{\OPT})+ \Phi(0) - \Phi(m)\leq 2r^{3/2}\sqrt{n}\cost(\text{MTF}_{\OPT}).
\end{eqnarray*}
\noindent The first inequality follows from inequality~(\ref{eq:technical}). The last inequality follows from the fact that $\Phi(0) \leq \Phi(m)$ since at round $0$ the permutations of MAE and $\text{MTF}_{\OPT}$ are the same. We complete the section with the proof of Lemma~\ref{l:move_to_front}, stating that $\cost(\text{MTF}_{\OPT}) \leq 2 \cdot \cost(\OPT_{\text{dynamic}})$.

\begin{lemma}\label{l:move_to_front}
For any sequence of requests $(S_1,\ldots, S_m)$, 
$$ \cost(\MTF_{\OPT}) \leq 2\cdot \cost(\OPT_{\text{dynamic}})$$
\end{lemma}

\begin{proof}
To simplify notation, let
$x_t$ and $y_t$ respectively denote the permutations of $\MTF_{\OPT}$ and $\OPT_{\text{dynamic}}$ at round $t$.
We will use as potential the function $\Phi(t) = \text{d}_{\text{KT}}(x_t, y_t)$ i.e.
the number of inverted pairs between the permutation $x_t$ and $y_t$.
Let $L_t$ denote the set of elements that are on the left of $e_t$ in permutation $x_t$ and on the left of $e_t$ in permutation $y_t$. Respectively, $R_t$ denotes the set of elements that are on the left of $e_t$ in permutation $x_t$, but on the right of $e_t$ in permutation $y_t$. Clearly, at round $t$, $\MTF_{\OPT}$ pays $L_t + R_t$ for moving cost and $1$ for accessing cost. At same time, $\OPT_{\text{dynamic}}$ pays at least $L_t + 1$ for accessing cost and $\text{d}_{\text{KT}}(y_t,y_{t-1})$ for moving cost.
\begin{eqnarray*}
\cost(\MTF_{\OPT}) + \Phi(t) - \Phi(t-1) &=& L_t + R_t + 1 + \text{d}_{\text{KT}}(x_t,y_t) - \text{d}_{\text{KT}}(x_{t-1},y_{t})\\
&+& \text{d}_{\text{KT}}(x_{t-1},y_t) - \text{d}_{\text{KT}}(x_{t-1},y_{t-1})\\
&\leq& L_t + R_t + 1 + (L_t - R_t)\\
&+& \text{d}_{\text{KT}}(x_{t-1},y_t) - \text{d}_{\text{KT}}(x_{t-1},y_{t-1})\\
&\leq& 2\cdot(L_t + 1) + \text{d}_{\text{KT}}(y_{t},y_{t-1})
\end{eqnarray*}
\noindent The first inequality follows by the definition of the of the sets $R_t$ and $L_t$, while the second by the triangle inequality in the Kendall-tau distance $\text{d}_{\text{KT}}(\cdot,\cdot)$.
The proof is completed by summing all over $m$ and using the fact that $\text{d}_{\text{KT}}(x_0,y_0) = 0$.
\end{proof}

\paragraph{Lower bound.} Next, we prove the following theorem showing a lower bound on $\MAE$, which nearly matches its upper bound. 
\repeattheorem{thm:mae-dyn-lb}

\noindent
The constructed  request sequence $(S_1, \ldots, S_m)$ will have the following properties:    
\begin{enumerate}
    \item The size of every request $S_t$ is $r \geq 3$.
    \item Every request forces the algorithm to pay an access cost of $\Omega(\sqrt{n})$. 
    \item The request sequence contains $\Omega(\sqrt{n})$ consecutive requests that share a common element, which we call the \emph{pivot}. 
\end{enumerate}

\noindent The first two properties ensure that $\MAE$ will pay $\Omega(r\sqrt{n})$ moving cost on every request. The third property will enable the optimal dynamic solution $\OPT_{\text{dynamic}}$ to pay $1$ access cost for the consecutive requests that share the common element.

Before proceeding to formal details of the proof of Theorem~\ref{thm:mae-dyn-lb}, we will present the main ideas of the construction. To this end, let $\pi_{0}$ be the MAE's initial permutation and think of it as being divided into $k+2$ blocks. The first block contains $k+r-1$ elements and all other blocks contain $k$ elements, therefore $n=k^2+2k+r-1$. The ``important'' elements, which the optimal dynamic solution uses to \emph{cover} all requests, are the elements of the second and last block initially. Meaning that the optimal dynamic solution can cover all requests using only the $2k$ ``important'' elements. 
Figure~\ref{fig:blocks} depicts the structure of the initial permutation for $k=3$ and $r=3$, where ``important'' elements are colored red and green.
\begin{figure}[h]
    \centering
   \includegraphics[scale=0.1]{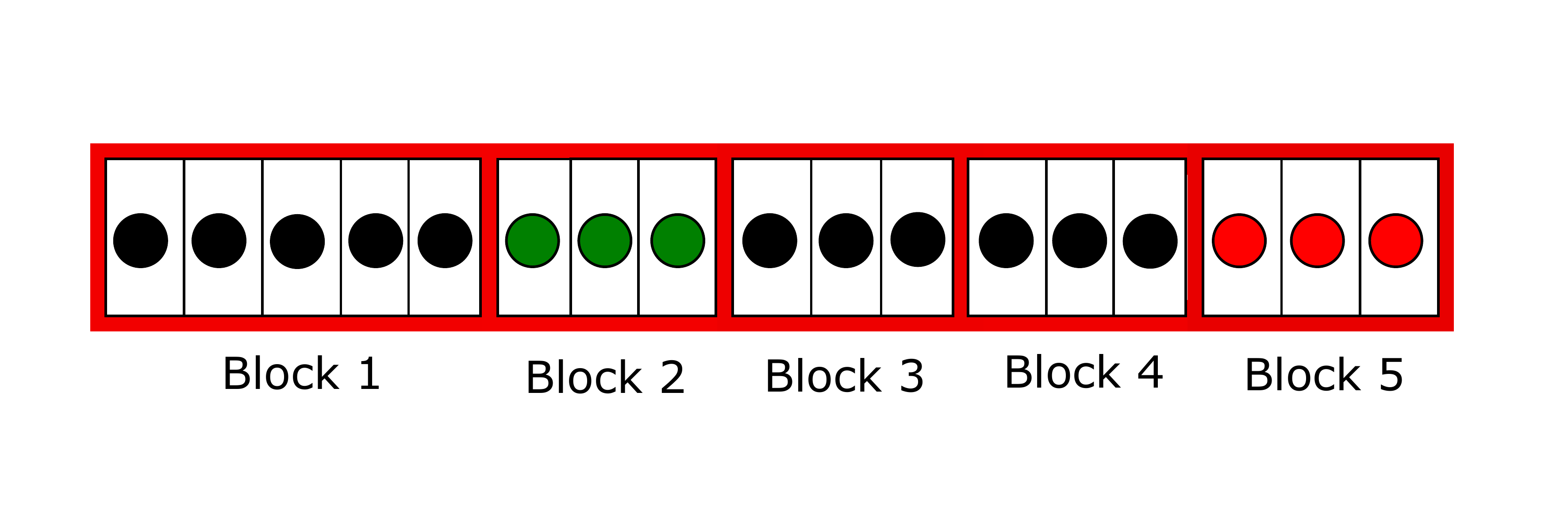}
    \caption{The structure in blocks for $k=3$ and $r=3$. Blocks are drawn with thick red borders. The elements that uses the optimal dynamic solution to cover the request sequence are colored green and red and lie in the second and last block initially .}
    \label{fig:blocks}
\end{figure}

\noindent The following definition formalizes the concept of blocks discussed above. Blocks are defined in terms of the indices of the permutation.
\begin{definition}\label{def:blocks}Let $\pi=(\pi_0, \ldots , \pi_m)$ be a sequence of permutations, where each permutation has size $n=k^2+2k+r-1$. We  divide each permutation into $k+2$ consecutive blocks $b_1, b_2, \ldots b_{k+2}$, where a block is a set of consecutive indices in each permutation with $|b_1|=k+r-1$ and $|b_2|=|b_3|= \ldots = |b_{k+2}|=k$. Therefore, $b_1=[1, \ldots , k+r-1]$ and $b_i=[(i-1)k+r, \ldots , ik+r-1]$, for $2 \leq i \leq k+2$.      
\end{definition}

For ease of exposition, we will divide the analysis in phases, which are subdivided in rounds. The first round starts with the request $S_1$, which contains the last element of the permutation $\pi_0$ ($pivot_1$) of the $\MAE$ algorithm and ends when $pivot_1$ arrives in position $k+r$. Inductively, at the start of a round $j$, $S_j$ contains the last element of permutation $\pi_{t}$ and lasts until this element arrives in position $k+r$. 
A phase  ends after $2k$ such rounds and a new phase begins. Formally:

\begin{definition}\label{def:phase} A phase  defines the period, starting with $k$ elements in block $b_2$ and $k$ elements in block $b_{k+2}$ and ending after $2k$ rounds of requests with these elements back in the blocks, where they started. A round $j$ starts when the adversary requests the last element , denoted as $pivot_j$, of the current permutation of MAE (with index $k^2+2k+r-1$) and ends when MAE places this element in the first position of the second block $b_2$ (with index $k+r$). The duration of the round is the number of requests from the start of the round until the end of the round and will be $k+1$.
\end{definition}

\begin{figure}[t]
    \centering
    \includegraphics[scale=0.1]{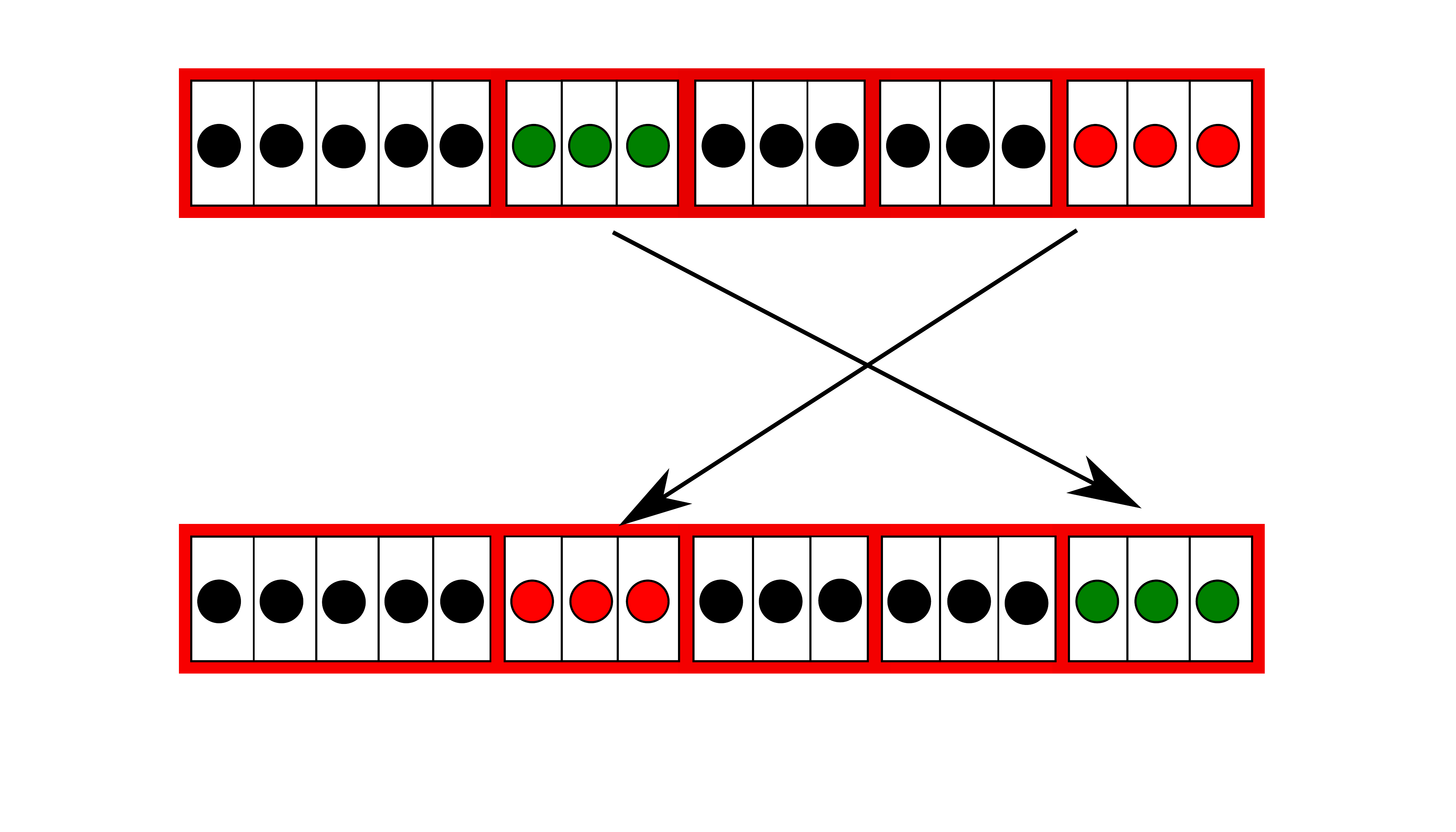}
    \caption{The permutation of $\MAE$ before the start of a round and after $k$ rounds, for $k=3$ and $r=3$. The elements of the second block have moved to the last block and vice versa.}
    \label{fig:symmetry}
\end{figure}

In order to describe the request sequence, we give the color green to $k$ elements and the color red to $k$ other. These are the elements, used by $\OPT_{\text{dynamic}}$ and every request contains one of them. The rest of them are colored black. Let all the green elements be in the second block ($b_2$) and all red elements be in the last block ($b_{k+2}$) initially. The request sequence may seem complicated, however its constructed following two basic principles. The first is to decrease the position of the $pivot$ element (which is initially red) by $k$ in the permutation and the second is to increase the position of the $k$ green  elements by one on every request. This is easily achieved by requests $S_t$ of the form:

 $$S_t = \{ \underbrace{\pi_t(k+1),\pi_t(k+2) \ldots ,\pi_t(k+r-2)}_{ black}, \underbrace{\text{element succeeding the green block}}_{black}, \underbrace{\text{pivot}}_{red}\} $$
\noindent
Then, at the end of round $j$, $pivot_j$  will be in the first position of the second block and the green elements will be in block $j+2$. After $k$ rounds, all red elements will be in the second block and all green elements will be in the last block, thus the adversary can repeat the request sequence starting again from round $1$.  Figure~\ref{fig:symmetry} depicts the positions of the ``important'' elements at the start of round $1$ and at the end of round $k$.

 The formal definition of the request sequence is shown in Algorithm~\ref{def:request-sequence}. 
\begin{algorithm}[h]
  \caption{Adversarial request sequence}\label{def:request-sequence}
Let $S_{ij}=[ p_1, \ldots , p_r ]$ be the $i$-th request of round $j$, where $p_1,\ldots, p_r$ denote the indices of the requested elements  in  the current permutation of MAE. The request sequence for $k$ consecutive rounds is (every round has $k+1$ requests):\\
\textbf{Round 1.} The $i$th request of round $1$, for $1\leq i< k$  is: \\
$S_{i1}=[k+1,\ldots, k+r-2, 2k+r-1+i, k^2+2k+r-1-(i-1)k]$\\
and the two last requests of round 1 are:\\
 $S_{k1}=[k+1,\ldots, k+r-1, 3k+r-1]$ and  $S_{(k+1)1}=[k,\ldots,k+r-2, 2k+r-1]$.\\
\textbf{Round $\mathbf{j}$ from 2 to k.}
The $i$th request of round $j$, for $i \notin \{ k-j+2, k+1 \} $ is : \\
$S_{ij}=[k+1, \ldots, k+r-2, (j+1)k+r-1+i, k^2+2k+r-1-(i-1)k]$,\\
for the request $i$ with $i= k-j+2$ is :\\
$S_{ij}=[k+1, \ldots, k+r-1, (j+1)k+r-1]$\\
and for $i=k+1$ is :\\
$S_{(k+1)j}=[k,\ldots, k+r-2 , 2k+r-1]$.\\
Then, the request sequence starts again from round 1.
\end{algorithm} 

 Notice that the last request of round $j$ places $pivot_j$ in the first position of block $b_2$, just in front of pivot elements of previous rounds. This way it is guaranteed that all $k$ red elements will be in $b_2$ after $k$ rounds. Moreover, at each round $j>1$ there is a request (the ($k-j+1$)th) that increases the positions of $k-j+2$ green elements by two, since they are passed from both the element succeeding them and the pivot element (the other $j-2$ elements are passed only by the  element succeeding the green block). Since, the adversary wants all green elements to be in block $j+2$ after round $j$, the next request does not move the $k-j+2$ green elements and the other $j-2$ are moved because  $pivot_j$ passes them (see Figure~\ref{fig:lower_bound}).

 The following lemma shows that $\MAE$ will arrive in a symmetric permutation after $k$ rounds.
 
\begin{lemma}\label{lem:symmetry}
Let $\pi$ be a permutation of definition \ref{def:blocks} at the start of a round  and let $X= \{ x_1,x_2, \ldots ,x_k \}$ and $Y= \{ y_1, y_2, \ldots ,y_k \}$ be the elements in $b_2$ and $b_{k+2}$ respectively before the start of a round. Then, after $k$ rounds, MAE has moved the elements of $Y$ in block $b_2$ and the elements of $X$ in block $b_{k+2}$. 
\end{lemma}

\begin{proof}
We have to prove that after $j$ rounds, $j$ elements of $Y$ will be in $b_2$ and all elements of $X$ will be in $b_{j+2}$. Observe that the only way to increase the index of an element $e$ in a permutation by $c$ is to move $c$ elements with higher index than $e$ in the permutation to positions with lower index than $e$ (it increases by one if an element arrives at the position of $e$).

Let $pivot_j$ be the pivot element of round $j$, which by definition is the last element in the permutation at the start of a round. After $k$ requests to $pivot_j$ involving also the ($k+1$)th element and elements in positions that do not increase the index of $pivot_j$, $pivot_j$ moves a total of $k^2$ positions to the left arriving in position $2k+r-1$ of the current permutation. Then, the last request of this round will move it to the first position of $b_2$, moving the pivot element of the previous phase to the second position of $b_2$. By construction of the request sequence, $pivot_j$ is the $j$th element of $Y$ requested so far. Therefore, at the end of round $j$, $j$ elements of $Y$ will be in the first $j$ positions of $b_2$ .

We now show that all elements of $X$ will be in $b_{j+2}$ at end of round $j$. Particularly, we show inductively that at the beginning of $j$th round all elements of $X$ are in block $b_{j+1}$ and  $\MAE$ moves all elements of $X$ to block $b_{j+2}$ when the $j$th round ends.

\textit{Induction Base:} In the first round, $k$ requests contain the element succeeding $X$, $r-2$ elements in positions $k+1, \ldots, k+r-2$  of the permutation and an element, which has index higher than $k+r-2$ (which is $pivot_1$) and does not change the positions of elements of $X$. The ($k+1$)th request does not change the positions of elements of $X$, since it contains elements with lower index. Therefore,
the algorithm moves all elements of $X$ exactly $k$ positions to the right and they will be in $b_3$ when round $1$ ends.

\textit{Inductive Step:} For $j>1$, we have $k-1$ requests, where each of them forces the algorithm to move elements of $X$ one position to the right. However, there is exactly one request $S_{ij}$ with $i=k-j+1$, where $k-j+2$  elements of $X$ increase their positions by two.  These elements are passed by both the element succeeding $X$ and the pivot element. The other $j-2$ elements are passed only by the element succeeding $X$, thus increasing their positions by one. Therefore, the next request of the adversary (the $k-j+2$-th) makes the algorithm move $pivot_j$  and the elements in positions $k+1, \ldots k+r-1$, therefore moving only the $j-2$ elements of $X$ one position to the right and the other elements of $X$ remain in their positions.

We conclude that at each round the elements of $X$ move to the next right block. Since, they are initially positioned in $b_2$, after $k$ rounds they end up in block $b_{k+2}$. 
\end{proof}

\begin{figure}[h]
\begin{subfigure}{.34\textwidth}
  \includegraphics[scale=0.07]{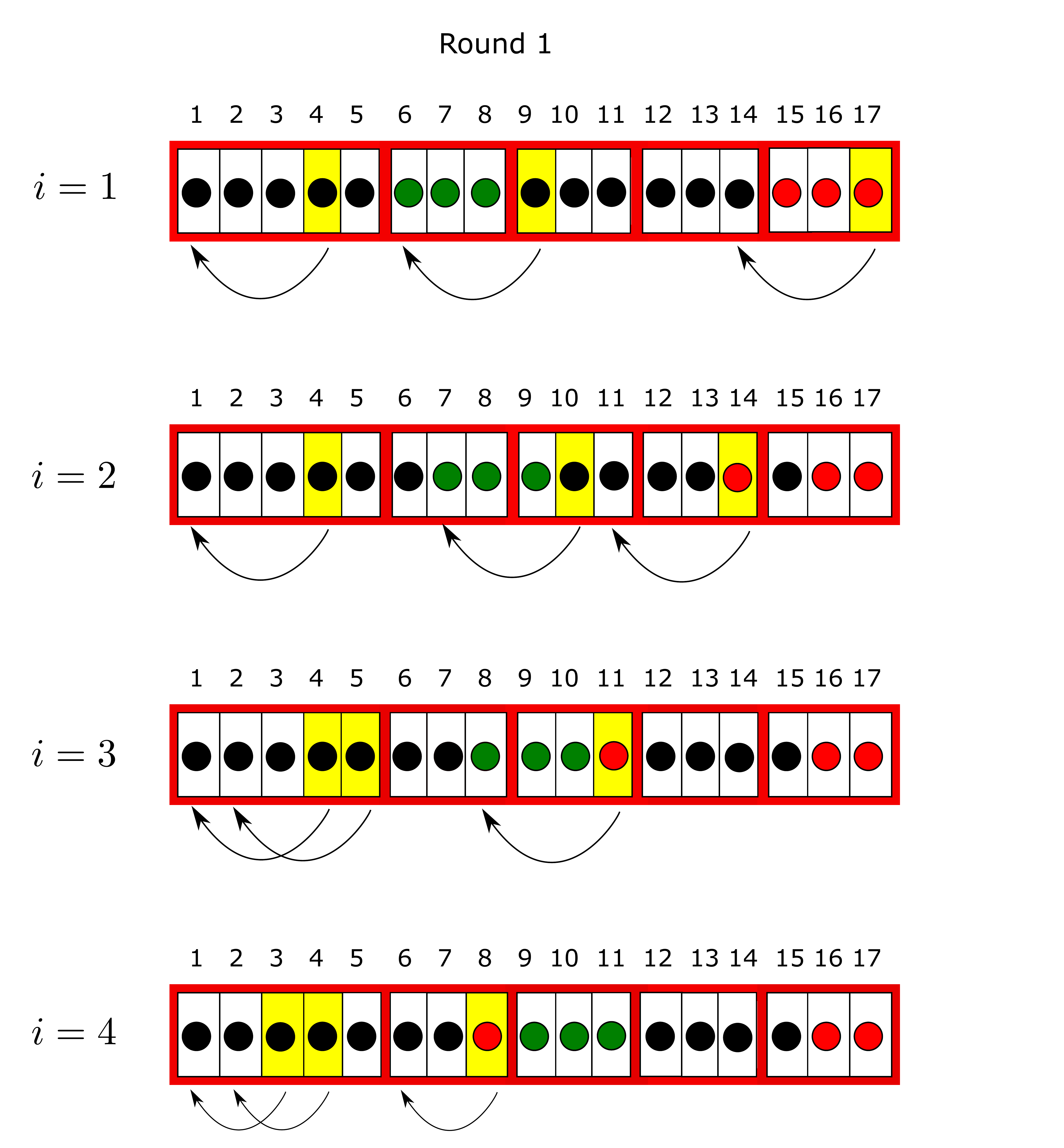}
\end{subfigure}%
\begin{subfigure}{.34\textwidth}
 \centering 
   \includegraphics[scale=0.07]{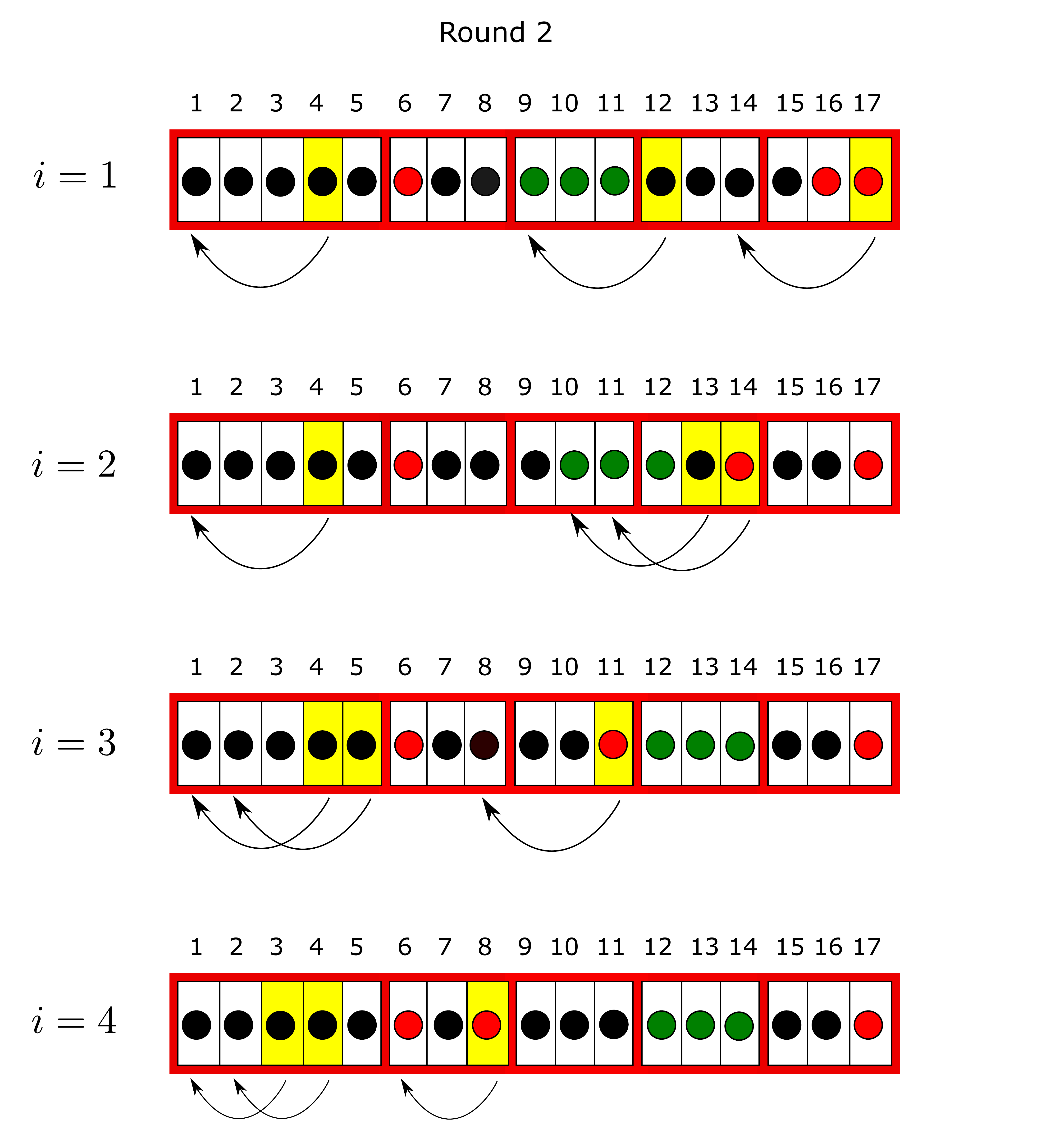}
\end{subfigure}%
\begin{subfigure}{.33\textwidth}
\includegraphics[scale=0.07]{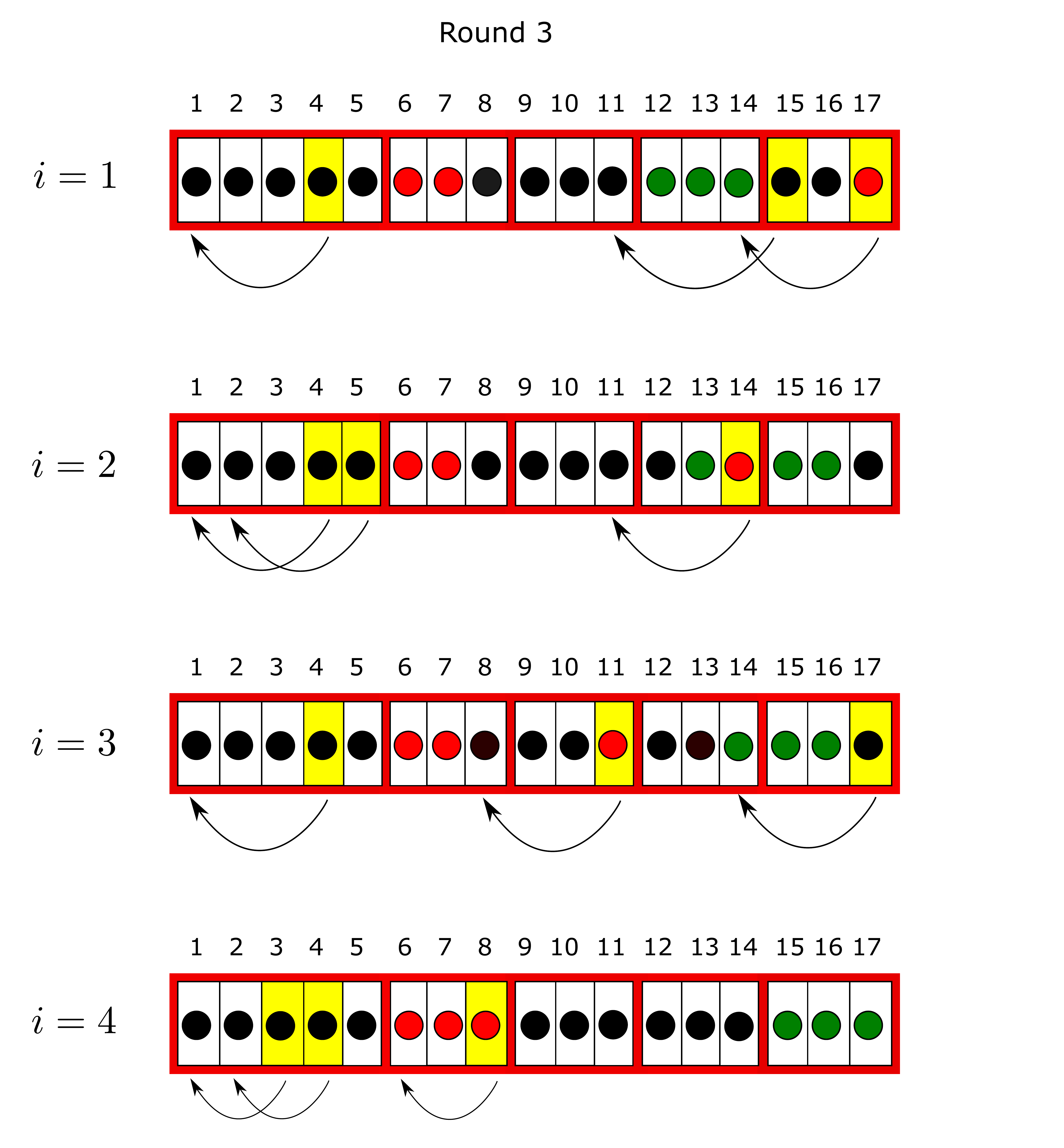}
\end{subfigure}%
\caption{Execution of $\MAE$ on the adversarial sequence for $r=3$ and $k=3$. The requested positions in the list have yellow background and the arrows indicate the positions of the elements after every request. For round $j>1$, the $i$th request with  $i=k-j+1$ makes $\MAE$ increase the position of $k-j+2$ green elements by two. The next request is such that only the $j-2$ remaining green elements increase their position by one.}
    \label{fig:lower_bound}
\end{figure}

\noindent
We are now ready to provide the proof of Theorem~\ref{thm:mae-dyn-lb}   
\repeattheorem{thm:mae-dyn-lb}
\begin{proof}
 We compute the costs paid by $\MAE$ and $\OPT_{\text{dynamic}}$  after $k$ rounds of a phase, since the request sequence of Definition \ref{def:request-sequence} is then repeated.

First, we bound the optimal cost. Initially, the optimal solution incurs a moving cost $O(k \cdot n)$ to move the $2k$ elements of $b_2$ and $b_{k+2}$ in the first $2k$ positions of its permutation.
Then, at the start of round $j$,  it brings $pivot_j$ to the first position and incurs an access cost of $1$ for $k+1$ consecutive requests. So, after $k$ rounds it pays at most $2k^2$ moving cost plus $k \cdot (k+1)$ access cost, which sums to at most $4k^2$. 

We now account the online cost. MAE pays $r \cdot k$ for the first $k$ requests of each round and $ r \cdot (k-1)$ for the last request. So, the total cost for $k$ rounds is $r\cdot k \cdot (k^2+k-1)>r \cdot k^3$. From Lemma~\ref{lem:symmetry}, all elements of $b_2$ are in $b_{k+2}$ and all elements of $b_{k+2}$ are in $b_2$ after $k$ rounds. 

The adversary can repeat the same strategy to create an arbitrarily long request sequence. Let $l$ be the number of times the same $k$-round strategy is applied. We get that  
$$\frac{\cost (\MAE)}{\cost(\OPT_{\text{dynamic}})}\geq \frac{l \cdot r \cdot k^3}{4l \cdot k^2+O(k \cdot n)}\to \Omega(k) \,.$$ The result follows for $l \to \infty $ and $k=\Omega(\sqrt{n})$.

\end{proof}

\end{onlyapp}

\section{Concluding Remarks}
\label{sec:concl}

Our work leaves several intriguing open questions. For the (static version of) Online MSSC, it would be interesting to determine the precise competitive ratio of the MAE algorithm; particularly whether it depends only on $r$ or some dependency on $n$ is really necessary. More generally, it would be interesting to determine the best possible performance of memoryless algorithms and investigate trade-offs between competitiveness and computational efficiency. 

For the online dynamic  MSSC problem, the obvious question is whether a $f(r)$-competitive algorithm is possible. Here, we showed that techniques developed for the list update problem seem to be too problem-specific and are not helpful in this direction. This calls for the use of more powerful and systematic approaches. For example, the online primal-dual method~\cite{BN09b} has been applied successfully for solving various fundamental problems~\cite{BBN12,BBMN15,BN06,BN13}. Unfortunately, we are not aware of a primal-dual algorithm even for the special case of list update; the only attempt we are aware of is in~\cite{tim16}, but this analysis basically recovers known (problem-specific) algorithms using dual-fitting. Our work gives further motivation for designing a primal-dual algorithm for list-update: this could be a starting point towards solving the online dynamic MSSC.  

In a broader context, the online MSSC is the first among a family of poorly understood online problems such as the multiple intents re-ranking problem described in Section~\ref{sec:related_work}. In this problem, when a set $S_t$ is requested, we need to cover it using $s \leq r$ elements; MSSC is the special case $s=1$. It is natural to expect that the lower bound of Theorem~\ref{thm:lb} can be generalized to $\Omega(r/s)$, i.e., as $s$ grows, we should be able to achieve a better competitive ratio. It will be interesting to investigate this and the applicability of our technique to obtain tight bounds for this problem.

\bibliographystyle{alpha}
\bibliography{references}

\newpage
\excludecomment{onlymain}
\includecomment{onlyapp}
\appendix

\end{document}